\documentclass[10pt,twocolumn,twoside]{IEEEtran}

\usepackage{color,graphicx,amsmath,amssymb,epsfig,mathrsfs,cite,multirow,array,amsthm}
\usepackage{wrapfig}
\usepackage{makecell}
\usepackage[caption=false,font=footnotesize]{subfig}
\usepackage[normalem]{ulem}

\usepackage[flushleft]{threeparttable}

\newcommand{\blue}[1]{{#1}} 

\IEEEoverridecommandlockouts

\newtheorem{theorem}{Theorem}[section]

\newtheorem{lemma}{Lemma}[section]

\everymath{\textstyle}
\let\displaystyle\textstyle

\newcommand{\squeezeup}{\vspace{-0.3cm}}

\begin{document}
\title{Identifying Seizure Onset Zone from the Causal Connectivity Inferred Using Directed Information}
\author{Rakesh Malladi,~\IEEEmembership{Student Member,~IEEE,} Giridhar Kalamangalam, Nitin Tandon, \\ and Behnaam Aazhang,~\IEEEmembership{Fellow,~IEEE}
\thanks{Copyright (c) 2016 IEEE. Personal use of this material is permitted. However, permission to use this material for any other purposes must be obtained from the IEEE by sending a request to pubs-permissions@ieee.org.}
\thanks{This work is funded in part by grant 1406447 from National Science Foundation and Texas Instruments. A portion of this work was presented at Society of Neuroscience (SfN) 2014 \cite{malladi2014} and International Conference on Acoustics, Speech and Signal Processing (ICASSP) 2015 \cite{malladi2015}.}
\thanks{Rakesh Malladi and Behnaam Aazhang are with the Department
of Electrical and Computer Engineering, Rice University, Houston, TX, 77005 USA.}
\thanks{Giridhar Kalamangalam and Nitin Tandon are with Department of Neurology and Department of Neurosurgery, respectively, at University of Texas Health Center, Houston, TX, 77005. E-mail: \{rm17, aaz\}@rice.edu, \{Giridhar.P.Kalamangalam, Nitin.Tandon\}@uth.tmc.edu.}
}

\maketitle

\begin{abstract}
In this paper, we developed a model-based and a data-driven estimator for directed information (DI) to infer the causal connectivity graph  between electrocorticographic (ECoG) signals recorded from brain and to identify the seizure onset zone (SOZ) in epileptic patients. Directed information, an information theoretic quantity, is a general metric to infer causal connectivity between time-series and is not restricted to a particular class of models unlike the popular metrics based on Granger causality or transfer entropy. The proposed estimators are shown to be almost surely convergent. Causal connectivity between ECoG electrodes in five epileptic patients is inferred using the proposed DI estimators, after validating their performance on simulated data. We then proposed a model-based and a data-driven SOZ identification algorithm to identify SOZ from the causal connectivity inferred using model-based and data-driven DI estimators respectively. The data-driven SOZ identification outperforms the model-based SOZ identification algorithm when benchmarked against visual analysis by neurologist, the current clinical gold standard. The causal connectivity analysis presented here is the first step towards developing novel non-surgical treatments for epilepsy.
\end{abstract}

\begin{IEEEkeywords}
Epilepsy, directed information, seizure onset zone, ECoG, causal connectivity.
\end{IEEEkeywords}

\section{Introdution}
Epilepsy is a common neurological disease affecting nearly $1\%$ of the world's population. Epilepsy is characterized by unprovoked seizures, which are periods of hypersynchronous activity in the brain. The current treatment options include medication, resective surgery and more recently, electrical stimulation approaches like vagus nerve and responsive neuro stimulation. However, medication is not able to stop seizures in about one-third of the patients. The efficacy of the other current neuro-modulation approaches is variable and almost never results in a cure \cite{bergey2015,rosenow2001}. The current approaches lack specificity and suffer from negative side effects (\cite{krook2015} and references therein). Selective modulation of the epileptic circuits in the brain via electrical stimulation \cite{sunderam2010}, optogenetics and designer receptive technologies \cite{krook2015} represent possible options for better treatments for this disabling disease. A crucial initial step in this endeavor is understanding how seizures originate and spread within the brain. Effective or causal connectivity \cite{Friston1994} quantifies how the activity spreads between different brain regions and can be used to characterize epileptic networks. In addition, causal connectivity can also be used to identify seizure onset zone (SOZ) (brain regions initiating seizures \cite{Luders2006}) and has been shown to predict the efficacy of resective surgery \cite{mierlo2014, pittau2015}. The main objective of this paper is to develop estimators of directed information (DI),  derive causal connectivity between brain regions and identify the SOZ using electrocorticographic (ECoG) data in patients with epilepsy.

Estimating causal connectivity from electrophysiological recordings of brain has been the focus of many papers. A good summary is provided in \cite{mierlo2014}. The causality referred to in this paper is in the Wiener-Granger causal sense \cite{bressler2011}. Metrics based on Granger causality (GC) \cite{Granger1969, Blinowska2011} and information theory like transfer entropy \cite{Schreiber2000} are commonly used to estimate causal connectivity between continuous-valued ECoG recordings. However, these techniques are well-suited only for a specific model and subset of the recorded signals. For instance, GC-based measures are applicable only for data from multivariate autoregressive (MVAR) processes. ECoG data are often modeled using linear MVAR model \cite{Blinowska2011, mierlo2014}, even though associations between seizure loci detected by ECoG recordings are likely nonlinear \cite{pereda2005, lehnertz2008}. 

We propose to develop a causal metric that would be applicable to diverse models and different types of electrophysiological recordings of brain. Directed information, used to infer causal connections between spike trains in \cite{Quinn2011a, so2012, soltani2014}, can indeed be further developed into a general technique to estimate causal connectivity. The definition of DI is based on the underlying probability distribution and no assumptions are imposed on the underlying distributions. Directed information was developed for discrete-valued time-series in \cite{Marko1973,Massey1990, Kramer1998} and nonparametrically estimated in \cite{Weissman2013a}. DI quantifies the amount of causal information about one time-series that is explained by the other time-series \cite{Quinn2011a}. Modified time-lagged directed information is proposed in \cite{Liu2012, Aviyente2012a} to reduce the computational complexity of estimating directed information. DI is also used in many other applications \cite{Rao2006, Tatikonda2009, Permuter2011}. The definition of DI is broadened to the class of continuous-valued processes like ECoG signals in this paper \cite{malladi2014,Liu2012,Amblard2011,Aviyente2012a}. If the data is assumed to be from a MVAR process with Gaussian white noise, DI is equivalent to Granger causality \cite{Amblard2011} and if the data satisfies the general Markov condition, DI is very closely related to transfer entropy \cite{Aviyente2012a,Liu2012}. The main advantage of DI over other existing techniques is that DI is applicable to a large class of electrophysiological recordings from brain, including spike trains, EEG and ECoG, and is not restricted to a particular class of models.

We developed an almost surely convergent model-based and data-driven estimators of DI in this paper, inspired by prior work \cite{Cover2006,Quinn2011a}. The performance of the proposed DI estimators was validated on linear and nonlinear simulated models and compared with the Granger causality metric \cite{Granger1969,barnett2014}. The statistical significance of the causal connection inferred using DI and GC estimates was demonstrated using an adaptation \cite{Diks2001} of stationary bootstrap \cite{Politis1994}. We then applied both model-based and data-driven DI estimators to infer causal connectivity graphs between ECoG channels from twelve seizures in five patients with epilepsy.

The DI metric with model-based and data-driven estimators proposed in this paper allows us the flexibility to simultaneously use both these estimators and identify which one leads to more reliable causal connectivity graphs from real ECoG data. This would also allows us to examine the appropriateness of imposing linear MVAR assumption on ECoG data. We used the model-based DI estimator with MVAR model assumption to detect the linear causal interactions and the data-driven DI estimator to detect both linear and nonlinear causal interactions between ECoG channels. We observed that nonlinear causal interactions between channels are stronger around the onset of a seizure, as widely believed \cite{lehnertz2008}.

We then proposed a model-based and a data-driven SOZ identification algorithm to identify SOZ from the causal connectivity graphs inferred using model-based and data-driven DI estimators respectively. The SOZ identified by model-based and data-driven algorithms are respectively the isolated nodes and strong sources in the corresponding causal connectivity graphs. Despite the numerous SOZ identification algorithms available \cite{Wilke2011,sabesan2009,panzica2013,mierlo2014, pittau2015}, the current clinical gold standard is still the visual analysis of ECoG data by the neurologist. We therefore compared the performance of both model-based and data-driven SOZ identification algorithms with visual analysis by the neurologist. We find that the data-driven approach outperforms the model-based approach and also leads to more interpretable results. The methodology proposed here should be extended to analyze the whole ECoG record to better understand the evolution of seizure mechanisms over time.

The main algorithmic contributions of this paper are 
\begin{itemize}
\item Developing an almost surely convergent model-based and data-driven DI estimator, described in sections III and IV.
\item Developing a MVAR model-based and a data-driven SOZ identification algorithm, described in section VI.
\end{itemize}
\section{Directed Information}
Consider the $N$ samples recorded at a sampling frequency $F_s$ from each ECoG electrode implanted in an epileptic patient. Without loss of generality, let us focus on two channels $\mathbf{X}$ and $\mathbf{Y}$. The $N$ samples recorded from the two channels are denoted by $\mathbf{X}^N = \left(x_1,x_2,\cdots,x_N \right)^{\text{T}}$ and $\mathbf{Y}^N = \left(y_1,y_2,\cdots,y_N \right)^{\text{T}}$. Also let $\mathbf{W}$ denote the matrix of samples recorded from a group of channels excluding $\mathbf{X}$ and $\mathbf{Y}$. For notational simplicity, the elements corresponding to the non-positive subscripts are treated as empty sets and the subscripts are not shown when equal to $1$. 

The directed information, $I\left(\mathbf{X}^N \rightarrow \mathbf{Y}^N\right)$, from $N$ samples of continuous-valued random process $\mathbf{X}$ to those of $\mathbf{Y}$  is defined as
\begin{equation}\label{dir_inf_def}
 I\left(\mathbf{X}^N \rightarrow \mathbf{Y}^N\right) = h\left(\mathbf{Y}^N\right) - h\left(\mathbf{Y}^N \| \mathbf{X}^N\right),
\end{equation}
where $h\left(\mathbf{Y}^N\right)$ is the differential entropy of the $N$-dimensional continuous random vector $\mathbf{Y}^N$ \cite{Cover2006} and $h\left(\mathbf{Y}^N \| \mathbf{X}^N\right)$ is the causally conditioned differential entropy of $\mathbf{Y}^N$ causally conditioned on $\mathbf{X}^N$. The causally conditioned differential entropy is defined as
\begin{equation} \label{cc_entropy}
h\left(\mathbf{Y}^N \| \mathbf{X}^N\right) = \textstyle\sum\limits_{n=1}^N h\left(y_n | \mathbf{Y}^{n-1},\mathbf{X}^n\right).
\end{equation}
The definitions of DI and causally conditioned differential entropy in \eqref{dir_inf_def} and \eqref{cc_entropy} are obtained by broadening the definitions of the same quantities from discrete-time, discrete-valued random processes \cite{Massey1990, Kramer1998} to discrete-time, continuous-valued processes \cite{malladi2015}. One of the main differences between discrete-valued and continuous-valued random processes is that the entropy of a discrete-valued process is always non-negative, whereas the  differential entropy of a continuous-valued process can be negative \cite{Cover2006}. However, DI is always non-negative since conditioning cannot increase differential entropy \cite{Cover2006}, i.e., $I\left(\mathbf{X}^N \rightarrow \mathbf{Y}^N \right) \geq 0$. DI can be interpreted as the number of bits of uncertainty in one process that is causally explained away by the other process. If $I\left(\mathbf{X}^N \rightarrow \mathbf{Y}^N \right) = 0$, then there is no causal influence from $\mathbf{X}$ to $\mathbf{Y}$. The DI is not a symmetric metric in general, i.e., $I\left(\mathbf{X}^N \rightarrow \mathbf{Y}^N\right)$ $\neq$ $I\left(\mathbf{Y}^N \rightarrow \mathbf{X}^N\right)$. Note that DI can also be expressed in terms of conditional mutual information \cite{Cover2006}, $I\left(y_n ; \mathbf{X}^n | \mathbf{Y}^{n-1} \right)$, as
\begin{align} \label{dir_inf_def1}
I\left(\mathbf{X}^N \rightarrow \mathbf{Y}^N\right) & = \textstyle\sum\limits_{n=1}^N \left\{h\left(y_n | \mathbf{Y}^{n-1}\right) - h\left(y_n | \mathbf{Y}^{n-1},\mathbf{X}^n\right) \right\} \nonumber \\
& = \textstyle\sum\limits_{n=1}^{N} I\left(y_n ; \mathbf{X}^n | \mathbf{Y}^{n-1} \right).
\end{align}
Now, the DI between the time-series $\mathbf{X}$ and $\mathbf{Y}$ is defined as
\begin{align} \label{dir_inf_def2}
I\big(\!\mathbf{X}\! \rightarrow \!\mathbf{Y}\!\big)\!&  =\!\! \lim_{\scriptscriptstyle N\rightarrow \infty}\! {\textstyle \frac{1}{N}} I\big(\!\mathbf{X}^N \!\!\rightarrow \!\mathbf{Y}^N\!\big)\!\! \nonumber \\
& = \!\! \lim_{\scriptscriptstyle N\rightarrow \infty} {\textstyle \frac{1}{N}}h\big(\!\mathbf{Y}^N\!\big)\!\! -\!\! \lim_{\scriptscriptstyle N\rightarrow \infty} {\textstyle \frac{1}{N}} h\big(\!\mathbf{Y}^N \| \mathbf{X}^N\!\big)\! \nonumber \\
& = \! h\big(\!\mathbf{Y}\!\big)\! - \!h\big(\!\mathbf{Y}\|\mathbf{X}\!\big),
\end{align}
provided the limits exist. $h\left(\mathbf{Y}\right)$ and $h\left(\mathbf{Y}\|\mathbf{X}\right)$ are respectively the differential entropy of $\mathbf{Y}$ and the causally conditioned differential entropy of $\mathbf{Y}$ given $\mathbf{X}$. The DI from $\mathbf{Y}$ to $\mathbf{X}$ is also similarly defined.

Furthermore, the DI defined earlier is easily extended to define directed information from $\mathbf{X}$ to $\mathbf{Y}$, causally conditioned on $\mathbf{W}$. Note that $\mathbf{W}$ comprises the samples recorded from a group of ECoG channels excluding $\mathbf{X}$ and $\mathbf{Y}$. The causally conditioned DI, $I\left(\!\mathbf{X}  \rightarrow \mathbf{Y} \|  \mathbf{W}\!\right)$, is defined as
\begin{align}\label{condDI_def}
I\left(\!\mathbf{X}  \rightarrow \mathbf{Y}\right. & \| \left. \mathbf{W}\!\right) \!  = \!\lim_{\scriptscriptstyle {N\rightarrow \infty}} {\textstyle\frac{1}{N}} I\left(\!\mathbf{X}^N \!\rightarrow \!\mathbf{Y}^N\!\| \mathbf{W}^N\!\right) \! \nonumber \\
&  =\! \lim_{\scriptscriptstyle {N\rightarrow \infty}} {\textstyle\frac{1}{N}} \left\{\! h\left(\!\mathbf{Y}^N\|\mathbf{W}^N\!\right)\! -\! h\left(\!\mathbf{Y}^N \! \| \!\mathbf{X}^N\!,\!\mathbf{W}^N\!\right)\!\right\}, \nonumber \\
& \blue{=\! h\left(\mathbf{Y}\|\mathbf{W}\right) - h\left(\mathbf{Y}\|\mathbf{W},\mathbf{X}\right)},
\end{align} 
where $h\left(\mathbf{Y}^N \| \mathbf{X}^N\!,\mathbf{W}^N\right)\! = \!\textstyle \sum\limits_{n=1}^N h\big(y_n|\mathbf{Y}^{n-1},\mathbf{X}^n,\mathbf{W}^n\big)$ is the  \blue{differential entropy of $\mathbf{Y}^N$ causally conditioned on $\mathbf{X}^N$ and $\mathbf{W}^N$, $h\left(\mathbf{Y}\|\mathbf{W}\right)$ is the causal conditioned differential entropy of $\mathbf{Y}$ given $\mathbf{W}$ and $h\left(\mathbf{Y}\|\mathbf{W},\mathbf{X}\right)$ is the causally conditioned differential entropy of $\mathbf{Y}$ given the causal past of $\mathbf{W}$ and $\mathbf{X}$.} We use the directed information defined here to learn the causal connectivity graph between all ECoG channels and identify the SOZ of epileptic patients.
\section{Universal Estimator for Directed Information} \label{sec:DI_Est_Algo}
A universal estimator for directed information between channels $\mathbf{X}$ and $\mathbf{Y}$, $I\left( \mathbf{X} \rightarrow \mathbf{Y} \right)$, and the causally conditioned DI, $I\left( \mathbf{X} \rightarrow \mathbf{Y} \| \mathbf{W}\right)$ is developed in this section. The proposed estimator is universal and is shown to be almost surely convergent assuming that the causal conditional likelihood (CCL) is known. If CCL is not known and is estimated, then the convergence of the proposed DI estimator is dependent on the CCL estimator. The ideas used in developing the proposed DI estimator are inspired by prior work \cite{Quinn2011a, Cover2006}. Without loss of generality, we will \blue{first} focus on estimating \blue{the pairwise DI,} $I\left(\mathbf{X} \rightarrow \mathbf{Y}\right)$. \blue{We will then outline the procedure to extend this pairwise DI estimator to estimate the causally conditioned DI, $I\left( \mathbf{X} \rightarrow \mathbf{Y} \| \mathbf{W}\right)$.} The inputs to the \blue{proposed pairwise DI} estimator are the observed $N$ samples of time-series $\mathbf{X}$ and $\mathbf{Y}$. The main idea is to develop an almost surely convergent estimator for the entropies in \eqref{dir_inf_def2} and the difference between the two entropy estimates is an almost surely convergent estimate for $I\left(\mathbf{X} \rightarrow \mathbf{Y}\right)$. Let us first focus on the causally conditioned differential entropy estimator, $\hat{h}\left(\mathbf{Y} \| \mathbf{X}\right)$. 

\textit{Assumption 1} - The random processes $\mathbf{X}$ and $\mathbf{Y}$ are assumed to be stationary, ergodic and Markovian in the observed time-window. These are reasonable assumptions to model ECoG data. First, an implicit assumption in the problem of estimating the causal connectivity from a ECoG data segment is that the causal connectivity does not vary in this segment, which is mathematically captured by stationarity. The entire ECoG data record is usually not stationary and stationary segments are identified using either sliding windows \cite{blinowska2006} or change-point detection algorithms \cite{malladi2013}. We used  the sliding window approach in this paper. \blue{A crucial parameter in this process is the length of the window in which data is assumed to be stationary. It is also important to realize that we need a minimum amount of data points to reliably estimate any unknown parameters involved. It is recommended that the number of data points should be much larger (as a thumb rule, at least an order of magnitude larger) than the number of parameters to be estimated \cite{blinowska2006}.} Directed information is then  estimated in each stationary segment using the algorithm proposed in this section. Directed information for the entire time-series is the sum of the DI estimates from each stationary segment and is interpreted as the total amount of uncertainty in one time-series in the entire recording window that is explained by the other time-series. Second, ergodicity is required to ensure that the estimates from long-enough recording windows converge to the true value. Finally, the Markovian assumption captures the dependence of the current activity on the past activity at different electrodes. Let the current sample of the time-series $\mathbf{Y}$ depend on the past $J_{yy}$ and past $K_{yx}$ samples of the time-series $\mathbf{Y}$ and $\mathbf{X}$ respectively. Note that $\left(J_{yy},K_{yx}\right)$ are unknown and should be estimated from data. The explicit model of the dependence is captured by the causal likelihood of $y_n$ conditioned on the past activity at electrodes $\mathbf{X}$ and $\mathbf{Y}$. This CCL is denoted by $\mathrm{P}\big(y_n| \mathbf{Y}^{n-1}_{n-J_{yy}},\mathbf{X}^n_{n-K_{yx}+1} \big)$ and can be estimated using either a model-based or a data-driven approach. Let us assume for now that CCL is known.

\textit{Assumption 2} - Let us also assume that differential entropy of the first sample, $y_1$, of time-series $\mathbf{Y}$ exists and that for some time-index $l \in \left[ 1,N\right]$, the conditional differential entropy of $y_l$, conditioned on $\mathbf{Y}_{l-J_{yy}}^{l-1}$ and $\mathbf{X}_{l-K_{yx}+1}^{l}$ also exists, i.e., $h\left(y_1\right),h\left(y_l|\mathbf{Y}_{l-J_{yy}}^{l-1}, \mathbf{X}_{l-K_{yx}+1}^{l} \right) \in \mathbb{R}$.
\begin{lemma} \label{lemma1}
Let Assumptions 1 and 2 hold. Then $h\left(\mathbf{Y}\right),$ $h\left(\mathbf{Y} \| \mathbf{X} \right),$  $I\left(\mathbf{X} \rightarrow \mathbf{Y}\right)$ exists and are in  $\mathbb{R}$.
\end{lemma}
\begin{proof} Stationarity and the property that conditioning cannot increase the differential entropy are the main ideas in the proof, which is in the Appendix~\ref{AppendixA}.
\end{proof}
\begin{lemma} \label{lemma2}
Let Assumptions 1 and 2 hold. Then for some time-index $l$
\begin{align} \label{lemma2eq}
\textstyle\frac{1}{N} h\left(\mathbf{Y}^N \| \mathbf{X}^N \right)\! = \!\mathbb{E}\!\left[\!-\! \log \mathrm{P}\!\left(y_l| \mathbf{Y}^{l-1}_{l-J_{yy}},\mathbf{X}^l_{l-\left(K_{yx}-1\right)} \!\right)\! \right].
\end{align}
\end{lemma}
\begin{proof}
The proof uses the definition of causally conditioned differential entropy \eqref{cc_entropy}, the Markovian and the stationarity assumptions. The proof is in the Appendix~\ref{AppendixA}.
\end{proof}
\begin{theorem} \label{theorem1}
Let Assumptions 1 and 2 hold. Then the almost surely convergent causally conditioned differential entropy estimator is
\begin{align} \label{theorem1eq} 
\hat{h}\left(\mathbf{Y} \| \mathbf{X} \right) \!=\! \textstyle\frac{1}{N} \!\textstyle\sum\limits_{n=1}^{N} \!\left\{\!- \!\log \mathrm{P}\!\left(\!y_n| \mathbf{Y}^{n-1}_{n-J_{yy}},\mathbf{X}^n_{n-\left(K_{yx}-1\right)}\! \right)\! \right\}. 
\end{align}
\end{theorem}
\begin{proof}
The proof is based on two observations: the first is that the right-hand side of \eqref{lemma2eq} does not depend on $N$ and therefore it is easy to compute its limit as $N\rightarrow \infty$. The second observation is that the strong law of large numbers (SLLN) for Markov chains \cite{Meyn2009} can be applied to estimate the expectation on the right-hand side of \eqref{lemma2eq}. The detailed proof is in the Appendix~\ref{AppendixA}.
\end{proof}
An almost surely convergent estimator for $h\left(\mathbf{Y}\right)$ can be easily derived using Theorem~\ref{theorem1}, simply by modeling the dependence of the current samples of $\mathbf{Y}$ on its own $J_{yy}^{\prime}$ past samples. This is equivalent to setting $K_{yx}=0$. The difference between the two estimators, $\hat{h}\left(\mathbf{Y}\right)$ and $\hat{h}\left(\mathbf{Y} \| \mathbf{X} \right)$, is the almost surely convergent estimator for DI from $\mathbf{X}$ to $\mathbf{Y}$, $\hat{I}\left(\mathbf{X} \rightarrow \mathbf{Y}\right)$. This is stated in Theorem~\ref{theorem2}.
\begin{theorem} \label{theorem2}
Let Assumptions 1 and 2 hold. The universal estimator for DI from time-series $\mathbf{X}$ to $\mathbf{Y}$ is 
\begin{equation} \label{theorem2eq}
\hat{I}\left(\mathbf{X} \rightarrow \mathbf{Y}\right) = \hat{h}\left(\mathbf{Y}\right) - \hat{h}\left(\mathbf{Y}\|\mathbf{X}\right) \xrightarrow{a.s.} I\left(\mathbf{X} \rightarrow \mathbf{Y}\right).
\end{equation}
\end{theorem}
\begin{proof} We have from Theorem~\ref{theorem1} $ \hat{I}\left(\mathbf{X} \rightarrow \mathbf{Y}\right) = \hat{h}\left(\mathbf{Y}\right) - \hat{h}\left(\mathbf{Y}\|\mathbf{X}\right) \xrightarrow{a.s.}  h\left(\mathbf{Y}\right) - h\left(\mathbf{Y}\|\mathbf{X}\right) = I\left(\mathbf{X} \rightarrow \mathbf{Y}\right).$
\end{proof}

The DI estimator in Theorem~\ref{theorem2} can be easily extended to estimate the causally conditioned directed information, $I\left( \mathbf{X} \rightarrow \mathbf{Y} \| \mathbf{W}\right)$. \blue{First, $h\left(\mathbf{Y}\|\mathbf{W}\right)$ is estimated using Theorem~\ref{theorem1}. We now need to estimate $h\left(\mathbf{Y}\|\mathbf{W},\mathbf{X} \right)$. Let $J_{yy}$, $K_{yw}$ and $K_{yx}$  respectively  denote the number of past samples of $\mathbf{Y}$, $\mathbf{W}$ and $\mathbf{X}$ that influence the current sample of $\mathbf{Y}$. Let us also assume the causal conditional likelihood $\mathrm{P}\left(y_n|\mathbf{Y}^{n-1}_{n-J_{yy}},\mathbf{W}^{n}_{n-K_{yw}+1},\mathbf{X}^{n}_{n-K_{yx}+1}\right)$ is known. A model-based and a data-driven approach to estimate this CCL is described in the subsequent section. Then Theorem.~\ref{theorem1} can be easily extended to show that
\begin{align} \label{CCentropy_est_eq} 
\hat{h}& \left(\mathbf{Y} \| \mathbf{W},\mathbf{X} \right) \nonumber \\
=&\textstyle\frac{1}{N}\!\! \textstyle\sum\limits_{n=1}^{N} \!\! \left\{\!-\!\log \mathrm{P}\!\! \left(y_n|\mathbf{Y}^{n-1}_{n-J_{yy}},\!\! \mathbf{W}^{n}_{n-K_{yw}+1},\! \mathbf{X}^{n}_{n-K_{yx}+1}\right) \right\}
\end{align}
is an almost surely convergent estimate of $h\left(\mathbf{Y}\|\mathbf{W},\mathbf{X} \right)$. From \eqref{condDI_def}, $\hat{I}\left( \mathbf{X} \rightarrow \mathbf{Y} \| \mathbf{W}\right)$ is the difference between the estimates, $\hat{h}\left(\mathbf{Y}\|\mathbf{W}\right)$ and $\hat{h}\left(\mathbf{Y}\|\mathbf{W},\mathbf{X} \right)$.} It is important to note that  as the number of channels included in $\mathbf{W}$ increases, the computational  complexity of the estimator also  increases.

The DI estimate, $\hat{I}\left(\mathbf{X} \rightarrow \mathbf{Y} \right)$, can be interpreted as the amount of causal information $\mathbf{X}$ contains about $\mathbf{Y}$. It is, however, important to note that $\hat{I}\left(\mathbf{X} \rightarrow \mathbf{Y} \right)$ is estimated from $N$ samples and is an estimate of the true value of DI from $\mathbf{X}$ to $\mathbf{Y}$. The statistical significance of the causal connection from $\mathbf{X}$ to $\mathbf{Y}$ inferred from $\hat{I}\left(\mathbf{X} \rightarrow \mathbf{Y} \right)$ is calculated using an adaptation \cite{Diks2001} of stationary bootstrap \cite{Politis1994}. $B$ stationary bootstrap samples of $\mathbf{X}$, denoted by $\mathbf{X}{\left(b\right)}$, are generated using the algorithm described in \cite{Diks2001}  for $b=1,2,\cdots,B$. The DI from $b^{th}$ stationary bootstrap sample $\mathbf{X}{\left(b\right)}$ to $\mathbf{Y}$, denoted by $\hat{I}\left(\mathbf{X}{\left(b\right)} \rightarrow \mathbf{Y} \right)$, is estimated using the proposed DI estimator. Note that there is no causal  influence from any of these bootstrap samples to $\mathbf{Y}$ by construction. Therefore the $B$ samples, $\hat{I}\left(\mathbf{X}{\left(b\right)} \rightarrow \mathbf{Y} \right)$, for $b=1,2,\cdots,B$ are from the null hypothesis of no causal influence. The statistical significance is determined by the P-value \cite{shi2008}. P-value is the probability that DI estimate greater than or equal to $\hat{I}\left(\mathbf{X} \rightarrow \mathbf{Y} \right)$ can be observed under the null hypothesis of no causality from $\mathbf{X}$ to $\mathbf{Y}$ and is computed from the empirical distribution of $\hat{I}\left(\mathbf{X}{\left(b\right)} \rightarrow \mathbf{Y} \right)$ for $b=1,2,\cdots,B$. If the P-value is less than a predetermined significance level $\delta$, the null hypothesis of no causal connection from $\mathbf{X}$ to $\mathbf{Y}$ is rejected. On the other hand, if the P-value is greater than $\delta$, the null hypothesis cannot be rejected and the causal connection from $\mathbf{X}$ to $\mathbf{Y}$ is not statistically significant. Note that the empirical distribution of $\hat{I}\left(\mathbf{X}{\left(b\right)} \rightarrow \mathbf{Y} \right)$ is concentrated around $0$, since the DI between time-series that are not causally connected is zero. Therefore, when the actual DI estimate is large enough, the P-value will be less than $\delta$ and the statistical significance assessment is not required. However, statistical significance assessment is useful when the DI estimate is close to zero. The significance assessment described here is applied to the simulated examples in section~\ref{sec:DI_Simulated_Data} to identify the significant causal connections, particularly useful when the DI estimates are close to zero. The above discussion assumes CCL is known. The likelihood, however, must be estimated from data in practice. A model-based and a data-driven approach to estimate CCL is described in the following section.
\section{Estimating Causal Conditional Likelihood} \label{sec:est_causal_likelihood}
\blue{Estimating DI from $\mathbf{X}$ to $\mathbf{Y}$ using the proposed DI estimator in section~\ref{sec:DI_Est_Algo} requires estimating two CCLs, $\mathrm{P}\left(y_n|\mathbf{Y}^{n-1}\right)$ and  $\mathrm{P}\left(y_n|\mathbf{Y}^{n-1},\mathbf{X}^n\right)$, while estimating DI from $\mathbf{Y}$ to $\mathbf{X}$ causally conditioned on $\mathbf{W}$  requires estimating two CCLs, $\mathrm{P}\left(y_n|\mathbf{Y}^{n-1},\mathbf{W}^{n}\right)$ and $\mathrm{P}\left(y_n|\mathbf{Y}^{n-1},\mathbf{W}^n,\mathbf{X}^n\right)$}. Let us focus on estimating $\mathrm{P}\left(y_n|\mathbf{Y}^{n-1},\mathbf{X}^n \right)$ for $n=1,2,\cdots,N$, which is required to estimate $\hat{h}\left(\mathbf{Y}\|\mathbf{X}\right)$. \blue{We will then describe how to extend this approach to estimate $\mathrm{P}\left(y_n|\mathbf{Y}^{n-1},\mathbf{W}^n,\mathbf{X}^n\right)$.} The CCLs are estimated using either model-based or data-driven techniques. The choice between model-based and data-driven approaches is determined by the application from which data is recorded. For instance, the time-series signals obtained from electrophysiological recordings of brain  or from stock markets are commonly modeled using MVAR models with Gaussian white noise. In this case, the CCL is easily estimated from the MVAR model of the data. Usually the parameters of the model are unknown and several classical techniques to estimate the unknown parameters are described in \cite{kay2010}. On the other hand, using model-based approaches to estimate CCLs from data recorded from nonlinear systems or systems without a prescribed linear model is non-trivial. This is because estimating the CCLs using model-based approach requires essentially inverting the nonlinear generative model, which is not trivial. Data-driven approaches do not have this limitation and are therefore preferred for nonlinear time-series data. A good review of the various data-driven algorithms that estimate probability distribution from data is provided in \cite{izenman1991, scott2015}. The model-based and the data-driven CCL algorithm used in this paper are described in the remainder of this section. 

\subsection{Model-based CCL Estimation} \label{subsec:param_causal_likelihood}
We will focus on estimating the CCL specifically for multivariate autoregressive process with Gaussian white noise in this paper. Let the time-series $\mathbf{X}$ and $\mathbf{Y}$ be sampled from such processes. Then, the samples of $\mathbf{Y}$ can be expressed as
\begin{align}\label{mvareq}
y_n \! = \! \textstyle\sum\limits_{j=1}^{J_{yy}} \alpha_j y_{n-j} \! + \! \textstyle\sum\limits_{k=1}^{K_{yx}} \beta_k x_{n-k+1} \! + \! z_n, \! n = \!1,2,\!\cdots,\! N,
\end{align}
where $z_n$ is the additive white Gaussian noise with zero mean and variance $\sigma_z^2$. Here $\alpha_j \: \text{for} \: j=1,2,\cdots,J_{yy}$ and $\beta_k \: \text{for} \: k=1,2,\cdots,K_{yx}$ are the parameters of the model and $J_{yy}$ and $K_{yx}$ are the model orders representing how many past samples of $\mathbf{Y}$ and of $\mathbf{X}$ respectively influence the current sample of $\mathbf{Y}$. It is easy to observe from \eqref{mvareq} that 
\begin{equation}\label{param_causal_likelihood}
\mathrm{P}\!\left(\!y_n|\mathbf{Y}^{n-1}\!,\!\mathbf{X}^n \!\right)\!\! \sim \!\! \mathcal{N}\!\Big(\!\textstyle\sum\limits_{j=1}^{J_{yy}} \! \alpha_j y_{n-j}\! +\!\! \textstyle\sum\limits_{k=1}^{K_{yx}} \! \beta_k x_{n-k+1},\sigma_z^2\! \Big).
\end{equation}

The two model orders, $J_{yy}$ and $K_{yx}$, and the parameter vector $\theta\left(J_{yy},K_{yx}\right) = \big(\alpha_1,\cdots,\alpha_{J_{yy}},\beta_1,\cdots,\beta_{K_{yx}},\sigma_z^2\big)^{\text{T}}$ are not known apriori and need to be estimated from the $N$ observed samples of $\mathbf{X}$ and $\mathbf{Y}$. The parameters and the model orders are estimated using a maximum likelihood (ML) estimator with  minimum description length \cite{Grunwald2007} penalty. ML estimator is known to be asymptotically consistent. Minimum description length is a model order selection procedure with good consistency properties \cite{Grunwald2007} and proportional to $\left(J_{yy}+K_{yx}\right)$. The optimal model orders $\big(\hat{J}_{yy},\hat{K}_{yx}\big)$ are the solutions of the following  problem:
\begin{align}
\textstyle \big(\hat{J}_{yy},\hat{K}_{yx}\big) =   \underset{\!\! \left(J_{yy},K_{yx}\right)}{\arg\min} \big\{\!\!-\!\!\textstyle\frac{1}{N} \!\log& \mathrm{P}\big(\mathbf{Y}^N \| \mathbf{X}^N \!;\! \hat{\theta}\big(J_{yy},K_{yx}\big)\big) \nonumber \\ 
& + \textstyle\frac{J_{yy} + K_{yx}}{2N}\log N \!\! \big\}, \label{modelorderesteq}
\end{align}
where $\hat{\theta}\left(J_{yy},K_{yx}\right)$ is the value of $\theta$ which minimizes the negative log-likelihood for a given $\big(J_{yy},K_{yx}\big)$ and is obtained by solving
\begin{equation}
\hat{\theta}\left(J_{yy},K_{yx}\right)\!\! =\!\! \underset{ \theta}{\arg\min}  \textstyle \!-\!\frac{1}{N}\! \log \!\mathrm{P}\!\left(\mathbf{Y}^N \!\|\! \mathbf{X}^N\! ;\! \theta\left(J_{yy}\!,\!K_{yx}\right)\! \right). \label{MLesteq}
\end{equation}
The ML estimation of $\theta$ for a given $\left(J_{yy},K_{yx}\right)$ in \eqref{MLesteq} is equivalent to the ML estimation of the parameters of a standard linear regression model \cite{kay2010}, since the CCL is Gaussian distributed \eqref{param_causal_likelihood}. The estimated parameters almost surely converge to the true parameter values \cite{Barron1991} resulting in almost surely convergence of the proposed DI estimator. The desired CCL is obtained by substituting the solutions of \eqref{MLesteq}, \eqref{modelorderesteq} in \eqref{param_causal_likelihood}. 
The resultant CCL is then substituted in \eqref{theorem1eq} to estimate $\hat{h}\left(\mathbf{Y}\| \mathbf{X} \right)$, which is further simplified to $\hat{h}\left(\mathbf{Y} \| \mathbf{X} \right) = \frac{1}{2} \log \left(2\pi e \hat{\sigma}^2_z\right)$, where $\hat{\sigma}^2_z$ is the estimate of the noise variance from \eqref{modelorderesteq}, \eqref{MLesteq}.

\blue{The MVAR model-based CCL estimation algorithm described above can be easily extended to estimate the CCLs required to estimate the causal conditional DI, $\hat{I}\left(\mathbf{X} \rightarrow \mathbf{Y}\|\mathbf{W}\right)$. Let us focus on estimating $\mathrm{P}\left(y_n|\mathbf{Y}^{n-1},\mathbf{W}^{n},\mathbf{X}^{n}\right)$, which is required to estimate $\hat{h}\left(\mathbf{Y}\|\mathbf{W},\mathbf{X}\right)$. Assuming MVAR model, let $J_{yy}, K_{yw}, K_{yx}$ respectively denote the number of past samples of $\mathbf{Y}$,$\mathbf{W}$,
$\mathbf{X}$ that influence $y_n$. Then for $n = 1,2,\cdots, N$, the current sample of $\mathbf{Y}$ can be expressed as
\begin{align}\label{cond_mvareq}
y_n \! \!= \!\! \textstyle\sum\limits_{j=1}^{J_{yy}} \alpha_j y_{n-j} \!\! + \!\! \textstyle\sum\limits_{k=1}^{K_{yw}} \gamma_k w_{n-k+1} \!\! + \!\! \textstyle\sum\limits_{l=1}^{K_{yx}} \beta_l x_{n-l+1} \!\! +  \!\! z_n.
\end{align}
The only difference with \eqref{mvareq} are the extra terms of the time-series $\mathbf{W}$. As a result, the CCL will still be Gaussian distributed with same variance as the distribution in \eqref{param_causal_likelihood} and whose mean contains the extra terms corresponding to the samples of $\mathbf{W}$. The unknown parameters under this model are $\alpha_j$,$\gamma_k$,$\beta_l$ for $j\!\!=\!\!1,\cdots,J_{yy}$, $k=1,\cdots,K_{yw}$,$l=1,\cdots,K_{yx}$ and the model orders $J_{yy}, K_{yw}, K_{yx}$. Maximum likelihood with minimum description length penalty can be used to estimate these parameters similarly. The resulting parameter estimates can then be used to calculate the CCL, which is substituted in \eqref{CCentropy_est_eq} to estimate $\hat{h}\left(\mathbf{Y}\|\mathbf{W},\mathbf{X}\right)$.}

\subsection{Data-driven CCL Estimation} \label{subsec:nonparam_causal_likelihood}
Let $J_{yy}$ and $K_{yx}$ denote the number of past samples of $\mathbf{Y}$ and $\mathbf{X}$ that influence the current sample of $\mathbf{Y}$. Then the CCL $\mathrm{P}\left(y_n|\mathbf{Y}^{n-1},\mathbf{X}^n\right)$ is same as $\mathrm{P}\left(y_n|\mathbf{Y}_{n-J_{yy}}^{n-1},\mathbf{X}_{n-K_{yx}+1}^n\right)$ and can be written as
\begin{align}\label{nonparameq}
\mathrm{P}\!\!\left(\!y_n|\mathbf{Y}_{n-J_{yy}}^{n-1}\!,\!\mathbf{X}_{n-K_{yx}+1}^n\!\right) \! \!=\! \! \frac{\!\mathrm{P}\left(\mathbf{Y}_{n-J_{yy}}^{n},\mathbf{X}_{n-K_{yx}+1}^n\right)}{ \!\mathrm{P}\!\left(\!\mathbf{Y}_{n-J_{yy}}^{n-1},\mathbf{X}_{n-K_{yx}+1}^n\!\right)\!}.
\end{align}

The joint distribution $\textstyle \mathrm{P}\left(\mathbf{Y}_{n-J_{yy}}^{n},\mathbf{X}_{n-K_{yx}+1}^n\right)$ of $J_{yy}+1$ and $K_{yx}$ consecutive samples of $\mathbf{Y}$ and $\mathbf{X}$ respectively is learned using kernel density estimator \cite{izenman1991} with Gaussian kernels. This estimator is implemented in the `ks' package in R \cite{duong2007}. The true values of $\left(J_{yy},K_{yx}\right)$ are not known and should be estimated. The joint density is learned for different values of $J_{yy}$ and $K_{yx}$ and the optimal values $\big(\hat{J}_{yy},\hat{K}_{yx}\big)$ are those that maximize the likelihood. The desired CCL is then estimated by substituting  $\textstyle \mathrm{P}\left(\mathbf{Y}_{n-\hat{J}_{yy}}^{n},\mathbf{X}_{n-\hat{K}_{yx}+1}^n\right)$ in \eqref{nonparameq}. The denominator in \eqref{nonparameq} marginalizes the joint distribution in numerator of \eqref{nonparameq} over $y_n$. This marginalization is implemented by approximating the integral with a Riemann sum of the distribution over a partition of the range of $y_n$. Note that the convergence of the estimated CCL to the true CCL depends on the underlying true data distribution \cite{wied2012}. $\hat{h}\left(\mathbf{Y}\|\mathbf{X}\right)$ is obtained by substituting the estimated  CCL in \eqref{theorem1eq}. 

\blue{The data-driven CCL estimation algorithm described above can be extended to estimate $\mathrm{P}\left(y_n|\mathbf{Y}^{n-1},\mathbf{W}^{n},\mathbf{X}^{n}\right)$ as well. Let $J_{yy}, K_{yw}, K_{yx}$ respectively denote the number of past samples of $\mathbf{Y},\mathbf{W}, \mathbf{X}$ that influence $y_n$. Then
\begin{align}\label{condDI_nonparameq}
\mathrm{P}\!\!\left(\!y_n\!|\!\mathbf{Y}^{n-1}\!,\!\mathbf{W}^{n}\!,\!\mathbf{X}^{n}\!\right) \!\! = \!\! \frac{\mathrm{P}\!\left(\!\mathbf{Y}_{n-J_{yy}}^{n}\!,\!\mathbf{W}_{n-K_{yw}+1}^n\!,\!\mathbf{X}_{n-K_{yx}+1}^n\!\right)}{\mathrm{P}\!\left(\!\mathbf{Y}_{n-J_{yy}}^{n-1}\!,\!\mathbf{W}_{n-K_{yw}+1}^n\!,\!\mathbf{X}_{n-K_{yx}+1}^n\!\right)}.
\end{align}
The joint distribution in the numerator can be similarly estimated using kernel density estimator \cite{izenman1991} with Gaussian kernels using `ks' package \cite{duong2007}. Note that the optimal values of the model-orders $J_{yy}, K_{yw}, K_{yx}$ are those that maximize the likelihood. The denominator in \eqref{condDI_nonparameq} is then obtained by marginalizing the distribution in the numerator similarly. The resultant numerator and denominator probabilities are substituted in \eqref{condDI_nonparameq} to estimate $\mathrm{P}\left(y_n|\mathbf{Y}^{n-1},\mathbf{W}^{n},\mathbf{X}^{n}\right)$, which is further substituted in \eqref{CCentropy_est_eq} to estimate $\hat{h}\left(\mathbf{Y}\|\mathbf{X},\mathbf{W}\right)$.}

\blue{The model-based and data-driven CCL algorithms described above can be easily modified to estimate $\mathrm{P}\left(y_n|\mathbf{Y}^{n-1}\right)$, which is required to estimate $\hat{h}\left(\mathbf{Y}\right)$. $\mathrm{P}\left(y_n|\mathbf{Y}^{n-1}\right)$ is obtained from either model-based or data-driven CCL by modeling the dependence of the current sample of $\mathbf{Y}$ just on its own past samples. $I\left(\mathbf{X} \rightarrow \mathbf{Y}\right)$ and $I\left(\mathbf{X} \rightarrow \mathbf{Y}\|\mathbf{W}\right)$ can now be estimated using the estimator proposed in section~\ref{sec:DI_Est_Algo}.} 

The DI estimator obtained by using the proposed estimator in Theorem.~\ref{theorem2} with model-based CCL and data-driven CCL estimation algorithms will henceforth be referred to as model-based and data-driven DI estimator respectively. If data is assumed to be drawn from MVAR model with Gaussian white noise, then model-based DI will be referred to as MVAR model-based DI estimator. Note that model-based approach is not restricted to just MVAR models, it is feasible for all those models from which we can  estimate the appropriate causal conditional likelihoods parametrically. We focused on MVAR with Gaussian white noise in this paper because ECoG is commonly modeled using this model in connectivity studies \cite{mierlo2014, Blinowska2011}. The performance of both the proposed DI estimators on simulated time-series data is demonstrated in the following section.
\section{Performance on Simulated Data} \label{sec:DI_Simulated_Data}
\blue{In this section, the performance of the proposed DI estimators is demonstrated using  simulated data generated from five models - two node bidirectional linear (section~\ref{subsec:linear_two_node}) and nonlinear (section~\ref{subsec:nonlinear_two_node}) causal network whose true connectivity is depicted in Fig.~\ref{Fig:Simulated_Network}a, a two node unidirectional noisy chaotic polynomial (section~\ref{subsec:chaotic_oscillator}) causal network whose true connectivity is shown in Fig.~\ref{Fig:Simulated_Network}b, four node linear (section~\ref{subsec:linear_four_node}) and nonlinear (section~\ref{subsec:nonlinear_four_node}) causal network whose true connectivity is depicted in Fig.~\ref{Fig:Simulated_Network}c.} A directed arrow in Fig.~\ref{Fig:Simulated_Network} represents a causal connection. \blue{The causal connection between two nodes, say from node $\mathbf{A}$ to $\mathbf{B}$ in Fig.~\ref{Fig:Simulated_Network}c, implies $I\left(\mathbf{A} \rightarrow \mathbf{B} \right) > 0$ or equivalently, that the past samples of $\mathbf{A}$ have some information about the current sample of $\mathbf{B}$}. We also compared the performance of the proposed DI estimators with the standard Granger causality (GC) \cite{Granger1969}. GC estimate is obtained from MVGC toolbox \cite{barnett2014}.  
Let us now describe the performance of the proposed DI estimators on the five models considered in detail.
\begin{figure}[!t]
\centering
\includegraphics[width=0.8\columnwidth]{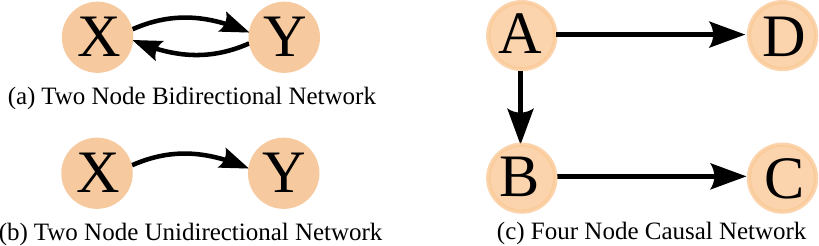}
\caption{The true causal connectivity graphs of the simulated data models used to validate the performance of the proposed model-based and data-driven DI estimators.}\label{Fig:Simulated_Network}
\squeezeup
\end{figure}

\subsection{Two Node Bidirectional Linear Causal Network}\label{subsec:linear_two_node}
Consider two time-series $\mathbf{X}$ and $\mathbf{Y}$ causally connected as shown in Fig.~\ref{Fig:Simulated_Network}a. The time-series $\mathbf{Y}$ is generated from 
\begin{equation} \label{linear_two_node_eq}
y_n = \beta_1 x_n + \beta_2 x_{n-1} + z_n, \: \text{for} \: n=1,2,\cdots,N,
\end{equation}
where $x_n$ and $z_n$ are sampled from an i.i.d Gaussian distribution with zero mean and variance $\sigma_x^2$, $\sigma_z^2$ respectively. The samples of $\mathbf{X}$ and $\mathbf{Z}$ are independent. The true value of the DI between $\mathbf{X}$ and $\mathbf{Y}$ in both directions is used to benchmark the performance of the proposed model-based and data-driven DI estimators. 

Let us first look at the true value of DI for the model  by \eqref{linear_two_node_eq} in two special cases. When $\beta_1=1, \beta_2=0$, \eqref{linear_two_node_eq} reduces to $y_n = x_n + z_n$, and it is obvious that both $\mathbf{X}$ and $\mathbf{Y}$ have equal causal information about each other. It is easy to see that $I\left(\mathbf{X} \rightarrow \mathbf{Y} \right) = I\left(\mathbf{Y} \rightarrow \mathbf{X} \right) = I\left(\mathbf{X} ; \mathbf{Y} \right) = C$, where $I\left(\mathbf{X};\mathbf{Y}\right)$ is the mutual information between $\mathbf{X}$ and $\mathbf{Y}$ and $C = \textstyle\frac{1}{2} \log \big(1+\textstyle\frac{\sigma_x^2}{\sigma_z^2} \big)$. The other special case occurs when $\beta_1=0,\beta_2=1$ and in this case \eqref{linear_two_node_eq} reduces to $y_n = x_{n-1}+z_n$. In this case, $\mathbf{X}$ has causal information about $\mathbf{Y}$, while $\mathbf{Y}$ has no causal information about $\mathbf{X}$. More precisely, $I\left(\mathbf{X} \rightarrow \mathbf{Y} \right) = I\left(\mathbf{X} ; \mathbf{Y} \right) = C$ and $I\left(\mathbf{Y} \rightarrow \mathbf{X} \right) = 0$. For the remaining case of non-zero $\beta_1, \beta_2$, the analytical expressions for DI are
\begin{align}\label{DIXtoY_YtoXeq}
I\big(\mathbf{X} \rightarrow \mathbf{Y}\big) \!\! = \!\!  \frac{1}{2}\!\! \log \!\!\left(\!\frac{|\beta_1 \beta_2|\sigma_x^2 }{\sigma_z^2} \!\right) \!\! & +\!\! \frac{1}{2} \!\!\cosh^{-1}\!\!\left(\!\!\frac{\left(\beta_1^2 + \beta_2^2\right) \sigma_x^2 + \sigma_z^2}{2|\beta_1 \beta_2| \sigma_x^2}\!\! \right)\!\!, \nonumber \\ 
I\big(\mathbf{Y} \rightarrow \mathbf{X}\big) =  \frac{1}{2} \log &\left(1+ \frac{\beta_1^2 \sigma_x^2 }{\sigma_z^2}\right)\!.
\end{align}
The derivation of \eqref{DIXtoY_YtoXeq} uses the tridiagonal matrix determinant from \cite{Hu1996} and is given in  Appendix~\ref{AppendixC}. Note from \eqref{DIXtoY_YtoXeq} that DI from $\mathbf{Y}$ to $\mathbf{X}$ does not depend on $\beta_2$. It is because the uncertainty in the current sample of $\mathbf{X}$ does not depend on $\beta_2$, when causally conditioned on the past of $\mathbf{X}$ and $\mathbf{Y}$.

\begin{figure}[!t]
\centering
\subfloat[$\beta_2 =\! \beta_1$]{
\includegraphics[width=0.45\columnwidth]{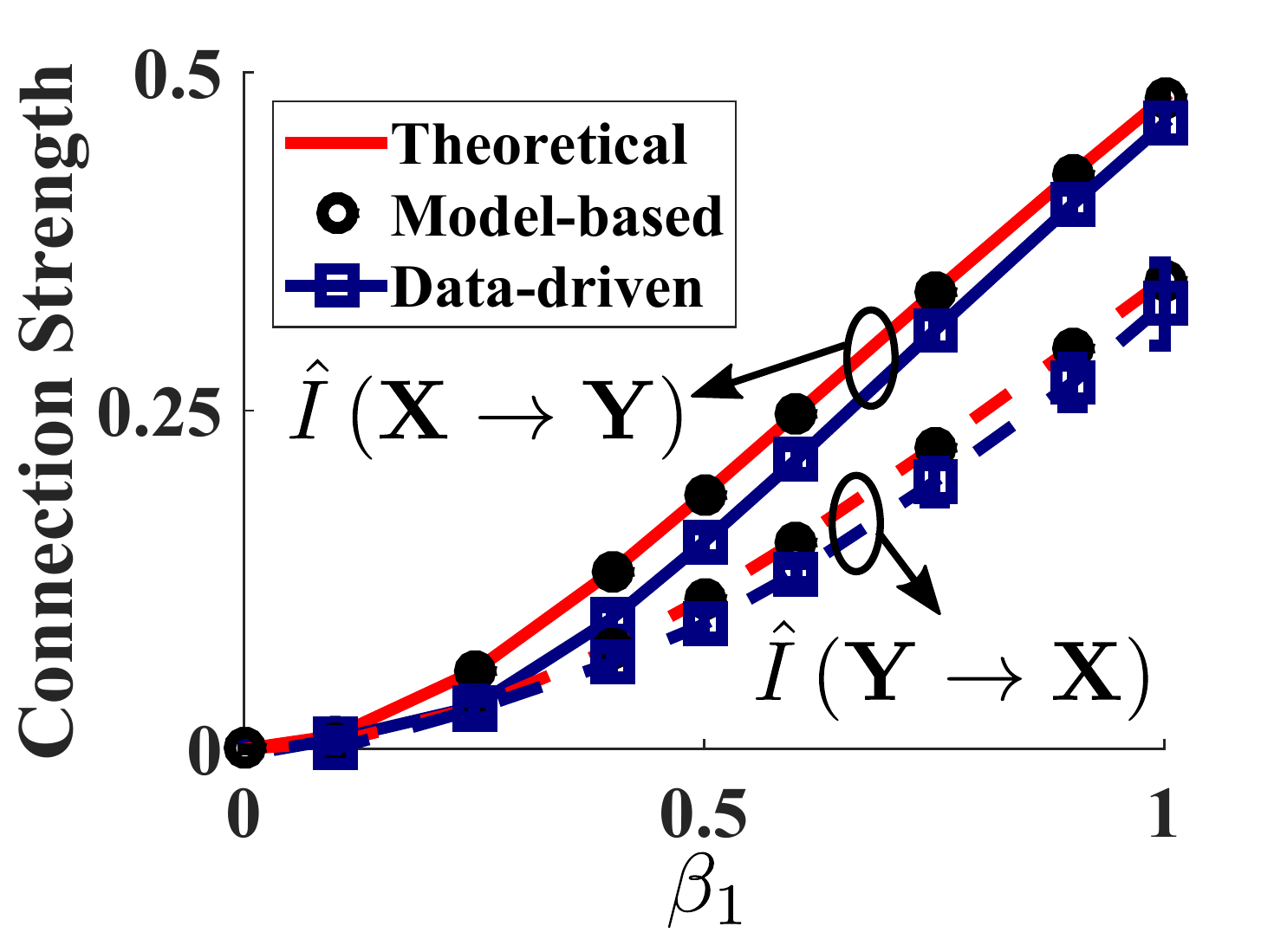}}
\hfill \vrule \hfill
\subfloat[$\beta_2 =\! 1 - \beta_1$]{
\includegraphics[width=0.45\columnwidth]{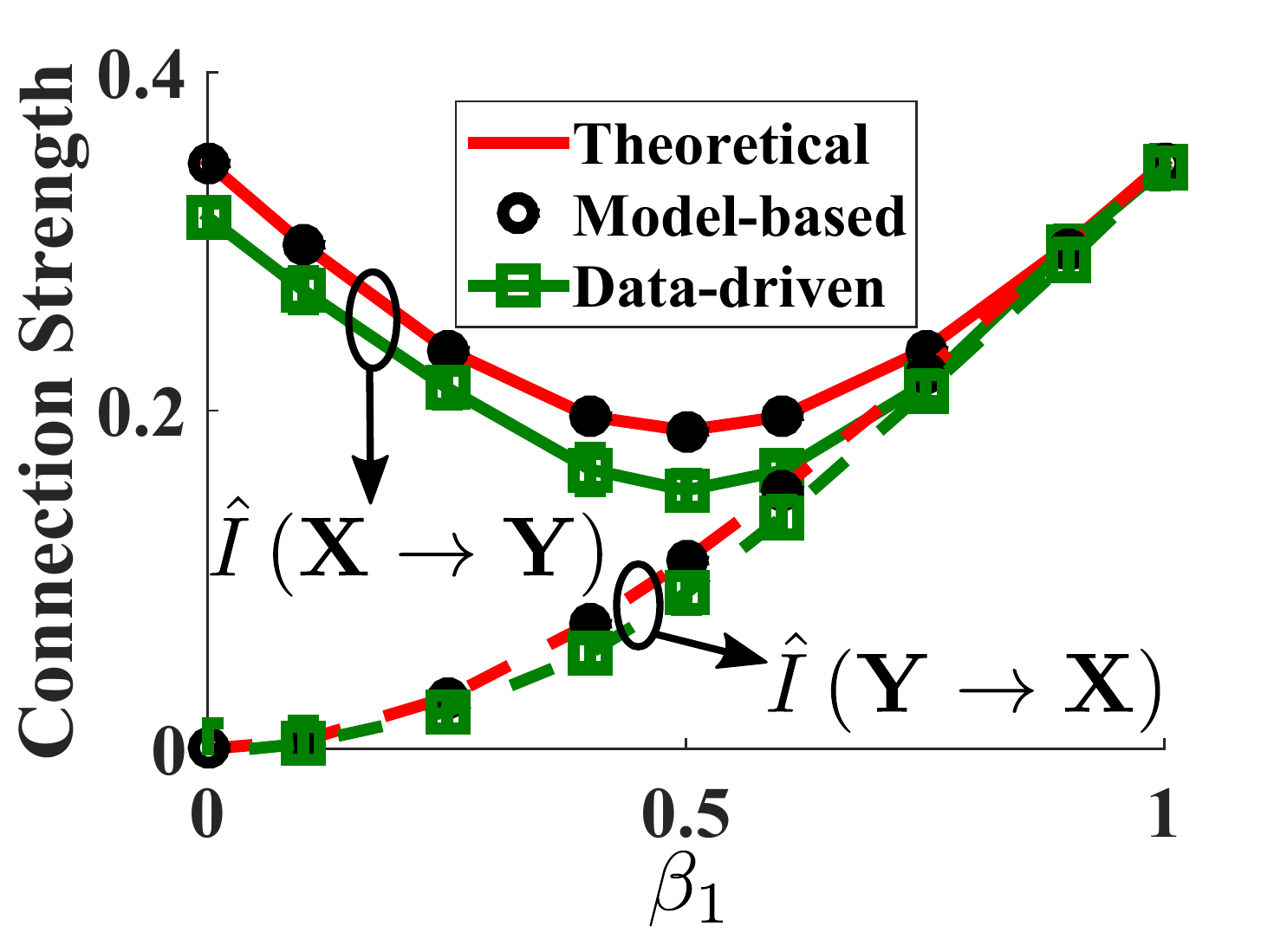} }
\caption{\blue{DI estimates and their standard deviation for the two node network (in Fig.~\ref{Fig:Simulated_Network}a) generated from a linear model \eqref{linear_two_node_eq} using analytical expression \eqref{DIXtoY_YtoXeq}, proposed model-based and data-driven DI estimators for different values of causal strength quantified by $\left(\beta_1,\beta_2\right)$. The DI estimates are plotted against $\beta_1$ with $\beta_2 = \beta_1$ in Fig.~\ref{Fig:Linear_Two_Node}a and with $\beta_2 = 1 - \beta_1$ in Fig.~\ref{Fig:Linear_Two_Node}b. 
}}\label{Fig:Linear_Two_Node}
\squeezeup
\end{figure}

The DI from $\mathbf{X}$ to $\mathbf{Y}$ and vice versa is estimated from $N = 10^5$ samples of $\mathbf{X}$ and $\mathbf{Y}$ generated with $\sigma_x^2 = 1$, $\sigma_z^2 = 1$ using the proposed model-based and data-driven DI estimators. The model-based DI estimator assumes that the time-series are modeled by a MVAR model with Gaussian white noise, whereas the data-driven CCL estimator does not impose any model assumptions on the data. Assuming $\mathbf{X}$, $\mathbf{Y}$ are from a MVAR process and when $x_n$ is included in the past samples of $\mathbf{X}$, Granger causality estimate from $\mathbf{X}$ to $\mathbf{Y}$ is equal to twice the MVAR model-based DI estimate from $\mathbf{X}$ to $\mathbf{Y}$ and vice versa \cite{Amblard2011}. We therefore do not show the GC estimates for linear MVAR models with Gaussian white noise. GC estimates are plotted only for nonlinear simulated models in this paper.

Fig.~\ref{Fig:Linear_Two_Node} plots directed information values obtained from the analytical expression in \eqref{DIXtoY_YtoXeq}, $\hat{I}\left(\mathbf{X} \rightarrow \mathbf{Y}\right)$ and $\hat{I}\left(\mathbf{Y} \rightarrow \mathbf{X}\right)$ from the proposed model-based and data-driven DI estimators for different values of $\beta_1 \in \left(0,1\right)$. The corresponding curves are respectively referred to as theoretical, model-based and data-driven. \blue{For the model-based and data-driven curves in Fig.~\ref{Fig:Linear_Two_Node}, multiple datasets of $\mathbf{X}$, $\mathbf{Y}$ are generated using different seeds for the random number generator. The mean and the standard deviation of the resultant estimates are plotted in Fig.~\ref{Fig:Linear_Two_Node}. The average standard deviation across all $\left(\beta_1,\beta_2 \right)$ in Fig.~\ref{Fig:Linear_Two_Node} is about $0.003$ and $0.01$ for the model-based and data-driven DI estimators respectively.} $\beta_2 = \beta_1$ in Fig.~\ref{Fig:Linear_Two_Node}a and $\beta_2 = 1-\beta_1$ in Fig.~\ref{Fig:Linear_Two_Node}b.

When $\beta_1=\beta_2$, a larger $\beta_1$ implies a stronger causal connection between $\mathbf{X}$ and $\mathbf{Y}$ and this should result in a larger DI. This expected trend is observed in Fig.~\ref{Fig:Linear_Two_Node}a. This implies that DI tracks the strength of the causal connection. Also in the corner case of $\beta_1 = \beta_2 = 0$, DI is zero in both directions as expected. In Fig.~\ref{Fig:Linear_Two_Node}b, DI estimates in the corner cases of $\beta_1 = 0$, $\beta_2 = 1$ and $\beta_1 = 1$, $\beta_2=0$ match with the analytical expression as expected. Also as $\beta_1$ increases from $0$ to $1$, the causal information $\mathbf{Y}$ has about $\mathbf{X}$ increases, and DI tracks this. This is demonstrated by observing that $\hat{I}\left(\mathbf{Y} \rightarrow \mathbf{X} \right)$ increases with $\beta_1$ in Fig.~\ref{Fig:Linear_Two_Node}b. Finally, it is clear from Fig.~\ref{Fig:Linear_Two_Node} that the model-based estimate matches the correct value of DI estimate from \eqref{DIXtoY_YtoXeq} and the data-driven estimator follows the true value of DI. This validates the accuracy of the proposed DI estimators. For this MVAR model with Gaussian white noise, the model-based DI estimator clearly performs better than the data-driven DI estimator and also has a lower run-time. We therefore use the MVAR model-based estimator to estimate DI between data modeled by MVAR processes with Gaussian white noise, instead of using the data-driven estimator. 

The adaptation of stationary bootstrap algorithm described earlier is used to assess the significance of the inferred causal connections for different values of  $\left(\beta_1,\beta_2\right)$. We observed that the null hypothesis of no causality from $\mathbf{Y}$ to $\mathbf{X}$ cannot be rejected for $\beta_1 \in \{0,0.1\}$ (P-value $>\delta = 0.05$) and can be rejected at all other points (P-value $< \delta$) in Fig.~\ref{Fig:Linear_Two_Node}. This is not surprising since $\hat{I}\left(\mathbf{Y} \rightarrow \mathbf{X} \right)$ is small for $\beta_1 \in \{0,0.1\}$ and hence did not result in a significant causal connection from $\mathbf{Y}$ to $\mathbf{X}$. Similarly, we observed that statistically significant causal connection from $\mathbf{X}$ to $\mathbf{Y}$ does not exist for $\beta_1=0, \beta_2=0$ (P-value $> \delta$) and exits at all other points (P-value $< \delta$) in Fig.~\ref{Fig:Linear_Two_Node}. This once again confirms our intuition that only large positive values of DI imply a statistically significant causal connection. 

\subsection{Two Node Bidirectional Nonlinear Causal Network} \label{subsec:nonlinear_two_node}
Now, consider time-series $\mathbf{X}$ and $\mathbf{Y}$ causally connected as shown in Fig.~\ref{Fig:Simulated_Network}a and are generated according to
\begin{equation} \label{nonlinearmodeleq}
y_n = \beta_1 x_n^2 + \beta_2 x_{n-1}^2 + z_n, \: \text{for} \: n = 1,2,\cdots,N,
\end{equation}
where $x_n$ and $z_n$ are sampled from an i.i.d Gaussian distribution with zero mean and variance $\sigma_x^2$, $\sigma_z^2$ respectively. Also, the samples of $\mathbf{X}$ and $\mathbf{Z}$ are independent. It is very non-trivial to estimate $\hat{I}\left(\mathbf{X}\rightarrow\mathbf{Y}\right)$ and $\hat{I}\left(\mathbf{Y} \rightarrow \mathbf{X} \right)$ using model-based DI estimator. This is because estimating $p\left(x_n|\mathbf{X}_1^{n-1},\mathbf{Y}_1^n\right)$  and $p\left(y_n|\mathbf{Y}_1^{n-1} \right)$ requires essentially inverting the non-linear, non-Gaussian generative model in \eqref{nonlinearmodeleq} and this is very hard even for this simple nonlinear model. These two probability densities are required to estimate $\hat{h}\left(\mathbf{X}\|\mathbf{Y} \right)$ and $\hat{h}\left(\mathbf{Y}\right)$ respectively. Therefore we only use the proposed data-driven DI estimator to estimate the DI from $\mathbf{X}$ to $\mathbf{Y}$ and vice versa. However, we can always assume that the data from the model in \eqref{nonlinearmodeleq} comes from a MVAR model with Gaussian noise, which is incorrect and estimate DI using the proposed MVAR model-based DI estimator. The resulting DI estimate will be half of the Granger causality estimate between these two time-series, $\hat{GC}\left(\mathbf{X}\rightarrow\mathbf{Y}\right)$ and $\hat{GC}\left(\mathbf{Y} \rightarrow \mathbf{X} \right)$. Note that GC also  assumes the data is generated from a MVAR process even though it is incorrect. We will now compare the performance of data-driven DI and GC estimates on this model.

Directed information and Granger causality between $\mathbf{X}$ and $\mathbf{Y}$ in both directions is estimated from $N=10^5$ samples generated with $\sigma_x^2 = 1$, $\sigma_z^2 = 1$ for different values of $\left(\beta_1,\beta_2\right)$ and plotted in Fig.~\ref{Fig:Nonlinear_Two_Node}. The DI and GC estimates are plotted for $\beta_2=\beta_1$ and $\beta_2 = 1 - \beta_1$ in Fig.~\ref{Fig:Nonlinear_Two_Node}a and Fig.~\ref{Fig:Nonlinear_Two_Node}b respectively. \blue{For each $\left(\beta_1, \beta_2\right)$, multiple datasets of $\mathbf{X}, \mathbf{Y}$ are generated with different random number generator seeds. The mean and the standard deviation of the resultant data-driven DI and GC  estimates are plotted in Fig.~\ref{Fig:Nonlinear_Two_Node}. The average standard deviation across all $\left(\beta_1, \beta_2\right)$ of the data-driven DI and GC estimates is $0.01$ and $1.8\times 10^{-5}$ respectively. In addition, the search space of the model order used by the Granger causality estimator is up to 20, i.e, $J_{yy},K_{yx} \in \left[1,20\right]$.} 
In Fig.~\ref{Fig:Nonlinear_Two_Node}a, $\hat{I}\left(\mathbf{X} \rightarrow \mathbf{Y}\right)$ increases with $\beta_1$ as expected. DI estimates also behave as expected in the corner cases of $\left(\beta_1,\beta_2\right) = \left(0,1\right) \: \text{and} \: \left(1,0\right)$ in Fig.~\ref{Fig:Nonlinear_Two_Node}b. $\hat{I}\left(\mathbf{Y} \rightarrow \mathbf{X} \right)$ increases with $\beta_1$ as expected. This once again demonstrates that DI tracks the strength of causal connections. On the other hand, Granger causality estimates in both directions are almost zero (of the order of $10^{-5}$), indicating that Granger causality cannot detect the causal connections in nonlinear models.

\begin{figure}
\centering
\subfloat[$\beta_2 =\! \beta_1$]{
\includegraphics[width=0.45\columnwidth]{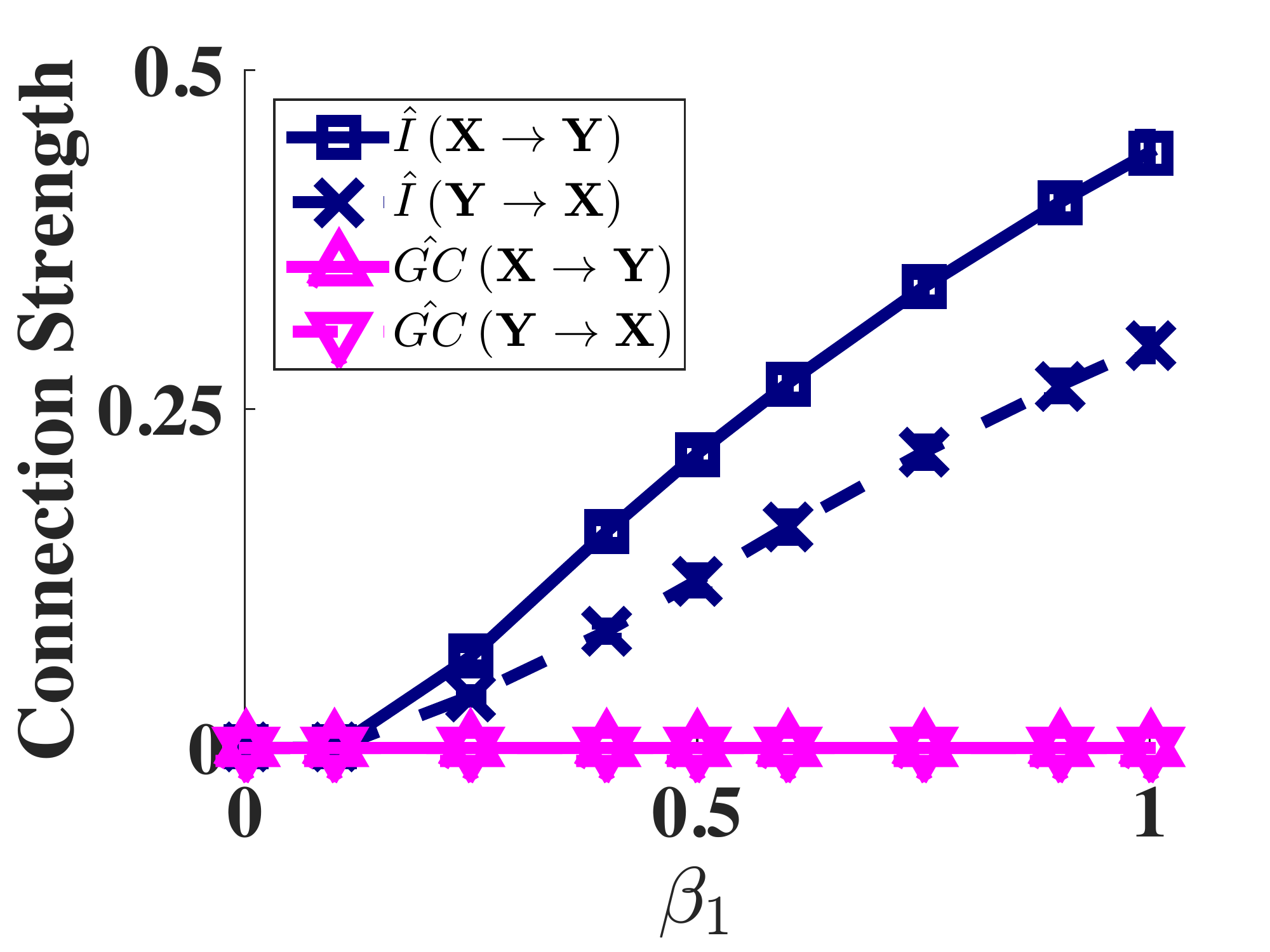}}
\hfill \vrule \hfill
\subfloat[$\beta_2 =\! 1 - \beta_1$]{
\includegraphics[width=0.45\columnwidth]{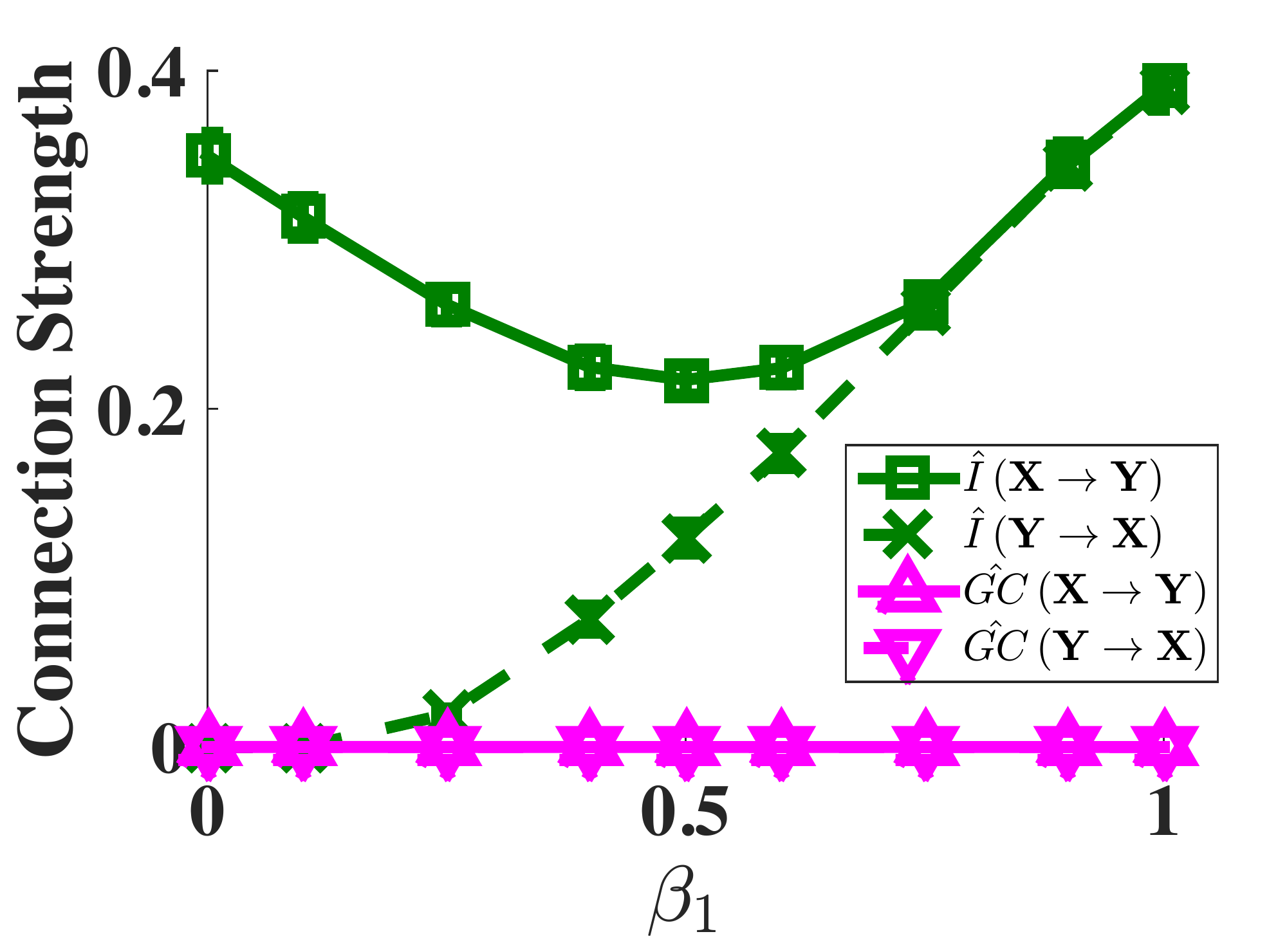} }
\caption{\blue{Data-driven DI and GC estimates, along with standard deviation of the estimates,  for the two node network (depicted in Fig.~\ref{Fig:Simulated_Network}a) generated from the nonlinear model \eqref{nonlinearmodeleq} for different values of causal strength quantified by $\left(\beta_1,\beta_2\right)$. The estimates are plotted against $\beta_1$ with $\beta_2 = \beta_1$ in Fig.~\ref{Fig:Nonlinear_Two_Node}a and with $\beta_2 = 1 - \beta_1$ in Fig.~\ref{Fig:Nonlinear_Two_Node}b. 
}}\label{Fig:Nonlinear_Two_Node}
\squeezeup
\end{figure}

The statistical significance of the inferred causal connections by DI and GC estimates for different values of  $\left(\beta_1,\beta_2\right)$ in Fig.~\ref{Fig:Nonlinear_Two_Node} is assessed using the stationary bootstrap algorithm described in section~\ref{sec:DI_Est_Algo}. Using DI, the null hypothesis of no causality from $\mathbf{Y}$ to $\mathbf{X}$ cannot be rejected for $\left(\beta_1,\beta_2\right) \in \{\left(0,0\right),\left(0,1\right)\left(0.1,0.1\right),\left(0.1,0.9\right)\}$ and from $\mathbf{X}$ to $\mathbf{Y}$ cannot be rejected for $\left(\beta_1,\beta_2\right) =\left(0,0\right)$ (P-value $> \delta = 0.05$) in Fig.~\ref{Fig:Nonlinear_Two_Node}. At all other points in Fig.~\ref{Fig:Nonlinear_Two_Node}, the null hypothesis of no causality can be rejected (P-value $< \delta$) using DI estimates. This once again confirms our intuition that large values of DI imply a statistically significant causal connection. For GC, the null hypothesis of no causality cannot be rejected at all points in Fig.~\ref{Fig:Nonlinear_Two_Node} implying that GC could not find statistically significant causal connections in nonlinear models. This example proves that DI is a more general causal connectivity metric that is not restricted to some particular models.

\begin{figure*}[!t]
\centering
\begin{minipage}[b]{0.24\textwidth}
\centering
\includegraphics[width = \textwidth]{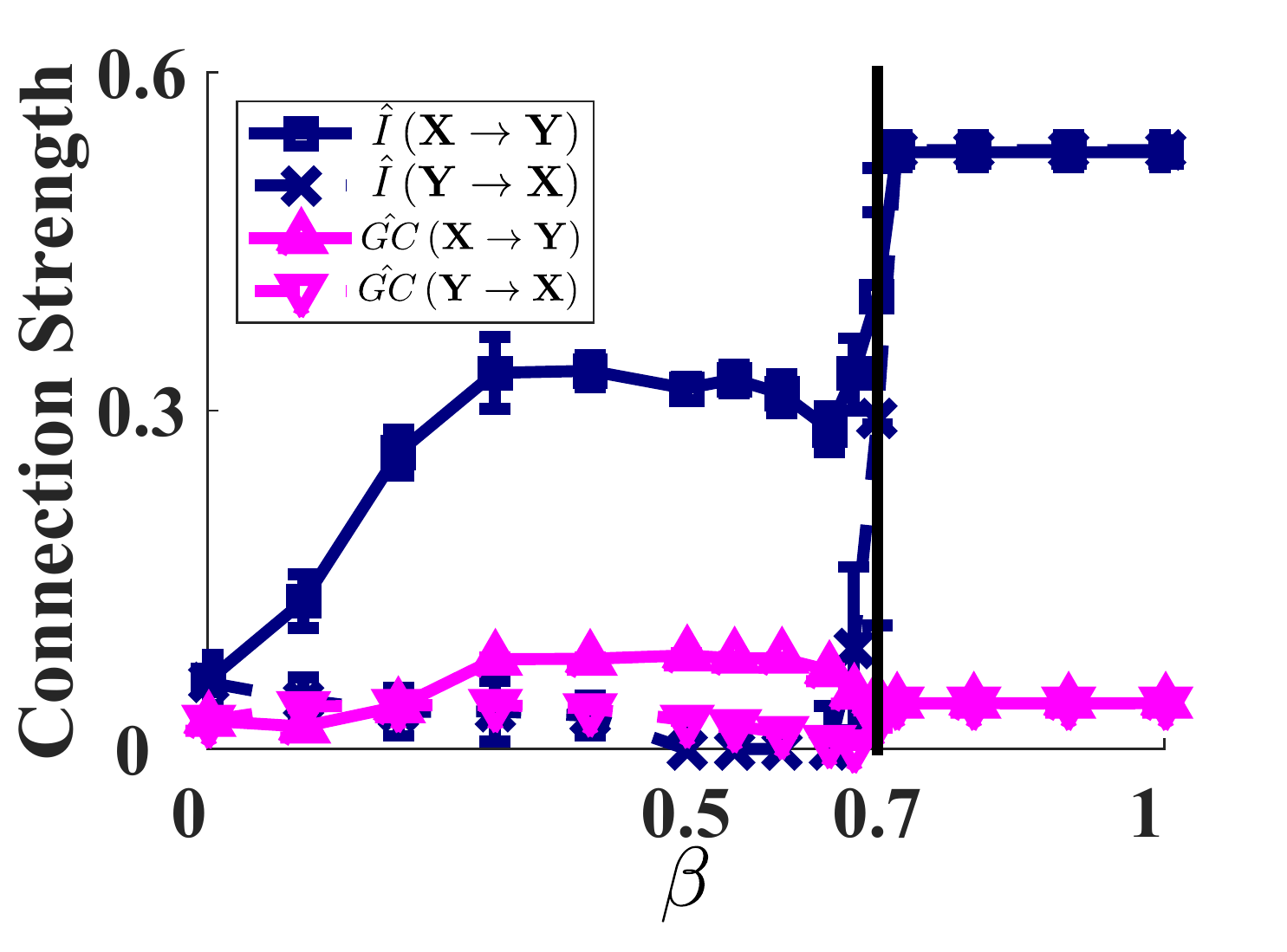}
\caption{\blue{Data-driven DI estimates and GC estimates, along with standard deviation of the estimates, for two node unidirectional network in Fig.~\ref{Fig:Simulated_Network}b generated from noisy chaotic polynomial map \eqref{chaotic_oscill} for different values of the coupling parameter $\beta$. 
}} 
\label{Fig:noisy_chaotic_poly_map}
\end{minipage}
\hfill \vrule \hfill
\begin{minipage}[b]{0.37\textwidth}
\centering
\subfloat[\blue{Model-based Estimator}]{
\includegraphics[width=0.45\columnwidth]{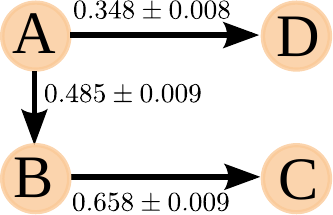}}
\hfill \vrule \hfill
\subfloat[\blue{Data-driven Estimator}]{
\includegraphics[width=0.45\columnwidth]{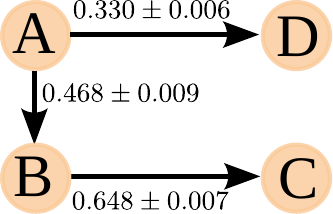} }
\caption{\blue{The causal network along with connection strengths between the four MVAR processes simulated from \eqref{linear_four_node_eq} estimated by the MVAR model-based DI and the data-driven DI estimators. The true causal connectivity graph between these four time-series is depicted in Fig.~\ref{Fig:Simulated_Network}c. It is clear that both DI estimators correctly infer the underlying causal network.}}\label{Fig:Linear_Four_Node}
\end{minipage}
\hfill \vrule \hfill
\begin{minipage}[b]{0.37\textwidth}
\centering
\subfloat[\blue{Model-based Estimator}]{
\includegraphics[width=0.45\columnwidth]{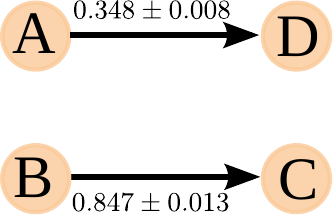}}
\hfill \vrule \hfill
\subfloat[\blue{Data-driven Estimator}]{
\includegraphics[width=0.45\columnwidth]{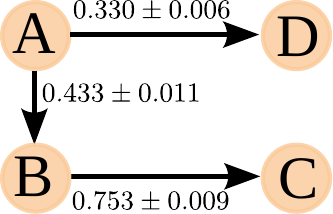} }
\caption{\blue{The causal network along with connection strengths between the four time-series simulated from \eqref{nonlinear_four_node_eq} estimated by the MVAR model-based DI and the data-driven DI estimators. The true causal connectivity graph between these four time-series is depicted in Fig.~\ref{Fig:Simulated_Network}c. It is clear that 
unlike MVAR model-based estimator, the data-driven estimator correctly infers the underlying causal connectivity graph.}}\label{Fig:Nonlinear_Four_Node}
\end{minipage}
\squeezeup
\end{figure*}

\subsection{Two Node Unidirectional Noisy Chaotic Polynomial Map}\label{subsec:chaotic_oscillator}
We now consider two unidirectionally coupled time-series $\mathbf{X}$ and $\mathbf{Y}$ whose underlying causal connectivity is shown in Fig.~\ref{Fig:Simulated_Network}b. The time-series $\mathbf{X}$ and $\mathbf{Y}$ are generated from a noisy chaotic polynomial map \cite{ishiguro2008} according to
\begin{align} \label{chaotic_oscill}
& x_n = 1.4 - x_{n-1}^2 + 0.3x_{n-2}, \nonumber \\
y_n = 1.4 - & \left(\beta x_{n-1} + \left(1- \beta \right)y_{n-1}\right)y_{n-1} + 0.3y_{n-2},
\end{align}
where $\beta$ controls the amount of causal information flowing from $\mathbf{X}$ to $\mathbf{Y}$. The initial two samples, $x_1,x_2,y_1,y_2$ are randomly chosen. The two time-series become completely synchronized for $\beta>0.7$. Gaussian i.i.d measurement noise of variance $0.01$ is added to both time-series $\mathbf{X}$ and $\mathbf{Y}$. For $\beta \in \left[0,0,7\right)$, strength of the causal connection from $\mathbf{X}$ to $\mathbf{Y}$ should increase with $\beta$ and there is no causal connection from $\mathbf{Y}$ to $\mathbf{X}$. For $\beta \in \left(0.7,1\right]$, since both time-series are completely synchronized and because of the measurement noise, there is a non-zero equally strong causal connection in both directions. In the absence of measurement noise for $\beta \in \left(0.7,1\right)$, $x_n=y_n$ leading to causal conditional entropy estimate of negative infinity and a DI estimate of infinity. The intuition behind this is that once the past of $\mathbf{X}$ is known, there is no uncertainty left in $\mathbf{Y}$. On the other hand, GC  estimates in the synchronized range will be close to zero because the past of $\mathbf{X}$ used by GC (unlike DI, GC does not include $x_n$ in the past of $\mathbf{X}$) does not contain any predictive information about $y_n$ resulting in a GC estimate of zero from $\mathbf{X}$ to $\mathbf{Y}$. Note that it is very non-trivial to apply model-based DI on this model because of the same reasons outlined in the previous simulated nonlinear model. We therefore only compare the performance of data-driven DI and GC estimates on this model.

DI and GC in both directions is estimated from $N=10^5$ samples of $\mathbf{X}$ and $\mathbf{Y}$ (after discarding the initial transient points) for different values of $\beta \in \left[0,1\right]$ and plotted in Fig.~\ref{Fig:noisy_chaotic_poly_map}. \blue{For each $\beta$, the time-series are generated from \eqref{chaotic_oscill} using different seeds of the random number generator. The mean and the standard deviation of the resulting data-driven DI and GC estimates are plotted in Fig.~\ref{Fig:noisy_chaotic_poly_map}. The average standard deviation across all $\beta$ for the data-driven DI and GC estimates is $0.03$ and $0.001$ respectively. The standard deviation was largest at $\beta = 0.7$, implying that it is very hard to estimate at the boundary before and after complete synchronization. In addition, the search space of the model order used by the Granger causality estimator is up to 20, i.e, $J_{yy},K_{yx} \in \left[1,20\right]$.} The DI estimate is obtained by subtracting two non-negative numbers and it can sometimes be a small negative number because of the inaccuracies in estimation algorithms or insufficient data or violation of stationarity assumptions \cite{Weissman2013a} and in those cases, we reset the DI estimate to be zero. For instance, the largest negative DI estimate we obtained for this model is $-0.06$ from $\mathbf{Y}$ to $\mathbf{X}$ at $\beta=0.6$ and we reset this estimate to $0$. It is clear from Fig.~\ref{Fig:noisy_chaotic_poly_map} that DI estimates behave as expected. DI from $\mathbf{X}$ to $\mathbf{Y}$ increase as $\beta$ goes from $0$ to $1$. On the other hand, the DI estimates from $\mathbf{Y}$ to $\mathbf{X}$ are very small numbers for $\beta<0.7$ and then there is a sudden jump in this estimate after $\beta>0.7$. This jump is because the time-series get synchronized for $\beta>0.7$. On the other hand, GC estimates in both directions are small positive numbers (when compared to DI estimates) for the whole range and become equal in value in the synchronized range of $\beta>0.7$. 

The statistical significance of the causal connections inferred by DI and GC estimates is assessed using the adaption of stationary bootstrap. The null hypothesis of no causality using DI estimates from $\mathbf{Y}$ to $\mathbf{X}$ cannot be rejected for $\beta<0.7$ and cannot be rejected for the connection from $\mathbf{X}$ to $\mathbf{Y}$ for $\beta<0.1$. This implies DI correctly identifies the presence of causal connection from $\mathbf{X}$ to $\mathbf{Y}$ for all $\beta \geq 0.1$ and the absence of causal connection from $\mathbf{Y}$ to $\mathbf{X}$ for $\beta<0.7$. It can also differentiate causally independent time-series ($\beta=0$) and completely identical time series ($\beta \in \left(0.7,1 \right]$). On the other hand, the null hypothesis of no causality cannot be rejected only for $\beta=0$ using GC estimates. This implies GC identifies the presence of a causal connection in both directions for all non-zero $\beta$, which is incorrect. This example also shows DI correctly infers causal connectivity from nonlinear models.

\subsection{\blue{Four Node Linear Causal Network}}\label{subsec:linear_four_node}
\blue{Now, consider the four node causal network depicted in Fig.~\ref{Fig:Simulated_Network}c. The four time series $\mathbf{A}$, $\mathbf{B}$, $\mathbf{C}$ and $\mathbf{D}$ are generated according to
\begin{eqnarray} \label{linear_four_node_eq}
b_n = a_{n-1} + a_{n-2} + z^{b}_n, & c_n = b_{n-1} + z^{c}_n,\nonumber \\
d_n = a_{n-2} + z^{d}_n, & \text{for} \quad n=1,2,\cdots,N, 
\end{eqnarray} 
where $a_n$, $z^{b}_n$, $z^{c}_n$ and $z^{d}_n$ are sampled from an i.i.d Gaussian distribution with zero mean and unit variance. In this network, $\mathbf{A}$ influences $\mathbf{C}$ indirectly through $\mathbf{B}$. This is an example of an `indirect' causal connection, in contrast with the connection from $\mathbf{A}$ to $\mathbf{B}$, which is a `direct' causal connection. DI estimate between pairs of time-series cannot differentiate between `direct' and `indirect' causal connections \cite{Quinn2011a}. For instance, the DI estimate from $\mathbf{A}$ to $\mathbf{C}$ is positive, even though $\mathbf{A}$ does not directly influence $\mathbf{C}$, but causally influences $\mathbf{C}$ via $\mathbf{B}$. A thorough discussion on the direct and indirect influences for point processes is in \cite{Quinn2011a} and is directly applicable here. Following the approach taken in \cite{Quinn2012, Amblard2011}, the `direct' causal influence from $\mathbf{A}$ to $\mathbf{C}$ is non-zero, if and only if $I\left(\mathbf{A}\rightarrow \mathbf{C} \| \mathbf{B}, \mathbf{D},\right) > 0$. However, estimating the causally conditioned DI when the number of channels recorded from is large (of the order of hundred's) is difficult because of the curse of dimensionality \cite{scott2015}. To overcome this, the pairwise DI is first estimated between all pairs of channels. The indirect influences are then resolved by first estimating only the required causal DI between two processes, conditioned on one more process. Then if required, the causal DI between two processes, conditioned on two more processes, is estimated and so on. The termination condition is determined by  the desired degree of `directness' in the inferred causal network. In this simulated example, we are interested in recovering the true `direct' causal network depicted in Fig.~\ref{Fig:Simulated_Network}c.}

To infer the true causal network, DI is estimated between these four time-series using both MVAR model-based and data-driven DI estimators. Model-based DI estimator assumes the data is generated from a linear causal MVAR model, whereas data-driven DI estimator does not impose any parametric model assumptions on the data. The data is generated from  \eqref{linear_four_node_eq} using $20$ different seeds to generate the Gaussian noise and the resultant estimates are averaged. We will first describe the performance using model-based DI estimator.

Model-based DI estimator is used to estimate the pairwise DI between all pairs of these four nodes, resulting a $4\times 4$ matrix with zeros on the diagonal. We found that $\hat{I}\left(\mathbf{A}\rightarrow \mathbf{B}\right)=0.485\pm0.009, \hat{I}\left(\mathbf{A}\rightarrow \mathbf{C}\right)=0.314\pm0.009 \: \text{and} \: \hat{I}\left(\mathbf{B}\rightarrow \mathbf{C}\right)=0.658\pm0.009$. To determine if there is an indirect causal connection from $\mathbf{A}$ to $\mathbf{C}$ or from  $\mathbf{B}$ to $\mathbf{C}$, we estimated $\hat{I}\left(\mathbf{A}\rightarrow \mathbf{C} \|\mathbf{B}\right)$ and $\hat{I}\left(\mathbf{B}\rightarrow \mathbf{C} \| \mathbf{A}\right)$ using the model-based causally conditioned DI estimator described in section~\ref{sec:DI_Est_Algo},~\ref{subsec:param_causal_likelihood}. We found that $\hat{I}\left(\mathbf{A}\rightarrow \mathbf{C} \|\mathbf{B}\right) = 0 $ and $\hat{I}\left(\mathbf{B}\rightarrow \mathbf{C} \| \mathbf{A}\right) = 0.344\pm0.009$. Therefore, $\mathbf{A}$ to $\mathbf{C}$ is an `indirect' connection via $\mathbf{B}$. Causally conditional DIs are estimated till the network is completely resolved and free of any indirect influences. The estimated causal network along with the strength and the standard deviation of the estimated causal connections is depicted in Fig.~\ref{Fig:Linear_Four_Node}a. It is clear from Fig.~\ref{Fig:Linear_Four_Node}a and Fig.~\ref{Fig:Simulated_Network}c that model-based DI estimator infers the true causal network correctly.

We now use the data-driven DI estimator to infer the true causal network. The pairwise DI is estimated between all pairs of these four nodes using the data-driven estimator, resulting in a $4\times 4$ matrix with zeros on the diagonal. Using this DI estimator, we find that $\hat{I}\left(\mathbf{A}\rightarrow \mathbf{B}\right)=0.468\pm0.009, \hat{I}\left(\mathbf{A}\rightarrow \mathbf{C}\right)=0.296\pm0.004 \: \text{and} \: \hat{I}\left(\mathbf{B}\rightarrow \mathbf{C}\right)=0.648\pm0.008$. To identify the presence of any indirect connections, we estimated $\hat{I}\left(\mathbf{A}\rightarrow \mathbf{C} \|\mathbf{B}\right)$ and $\hat{I}\left(\mathbf{B}\rightarrow \mathbf{C} \| \mathbf{A}\right)$ using the model-based causally conditioned DI estimator described in section~\ref{sec:DI_Est_Algo},~\ref{subsec:nonparam_causal_likelihood}. We found that $\hat{I}\left(\mathbf{A}\rightarrow \mathbf{C} \|\mathbf{B}\right) = 0 $ and $\hat{I}\left(\mathbf{B}\rightarrow \mathbf{C} \| \mathbf{A}\right) = 0.273\pm0.009$. Therefore, $\mathbf{A}$ to $\mathbf{C}$ is an `indirect' connection via $\mathbf{B}$. This procedure is continued to identify and remove all indirect causal connections. The resultant estimated direct causal network is depicted in Fig.~\ref{Fig:Linear_Four_Node}b. It is clear that data-driven DI also recovers the true network correctly. Moreover, it is clear from Fig.~\ref{Fig:Linear_Four_Node} that for this model, both model-based and data-driven DI estimators correctly infer the underlying causal network, which is not surprising since the underlying model is a linear MVAR model.

\subsection{\blue{Four Node Nonlinear Causal Network}} \label{subsec:nonlinear_four_node}
\blue{We now use a nonlinear model to generate the four time-series $\mathbf{A}$, $\mathbf{B}$, $\mathbf{C}$ and $\mathbf{D}$ whose underlying causal connectivity graph is depicted in Fig.~\ref{Fig:Simulated_Network}c. $N$ samples from the four time-series are generated according to
\begin{eqnarray} \label{nonlinear_four_node_eq}
b_n = a_{n-1}^2 + a_{n-2}^2 + z^{b}_n, & c_n = b_{n-1} + z^{c}_n,\nonumber \\
d_n = a_{n-2} + z^{d}_n, & \text{for} \quad n=1,2,\cdots,N,
\end{eqnarray} 
where $a_n$, $z^{b}_n$, $z^{c}_n$ and $z^{d}_n$ are sampled from an i.i.d Gaussian distribution with zero mean and unit variance. The only difference with the model in section~\ref{subsec:linear_four_node} is that the causal connection from $\mathbf{A}$ to $\mathbf{B}$ is now nonlinear.} 


\blue{First, we infer the true causal connectivity for this model using the MVAR model-based DI estimator. This DI estimator assumes that the data is drawn from a linear MVAR model, which is not true for this model. It is clear from \eqref{nonlinear_four_node_eq} that the time-series $\mathbf{B}$ is not generated from a linear MVAR model. Pairwise DI is estimated using this model between all pairs of these four time-series resulting in a $4\times 4$ matrix with zeros on the diagonal. The only significant causal connections estimated by the model-based DI estimator are from $\mathbf{B}$ to $\mathbf{C}$ and from $\mathbf{A}$ to $\mathbf{D}$. This process is repeated for data generated using 20 different seeds and the resultant DI estimates are averaged. We find that $\hat{I}\left(\mathbf{B} \rightarrow \mathbf{C}\right)=0.847 \pm 0.013$ and $\hat{I}\left(\mathbf{A} \rightarrow \mathbf{D}\right)=0.348 \pm 0.008$. It is also clear that there are no indirect connections to resolve in this case. The underlying causal connectivity graph estimated by the model-based DI estimator is depicted in Fig.~\ref{Fig:Nonlinear_Four_Node}a. It is clear from this figure that model-based DI estimator could not recover this true network correctly. This is not surprising since the MVAR model-based estimator can only identify linear causal connections and cannot identify the nonlinear causal connections. As result, the connection from $\mathbf{A}$ to $\mathbf{B}$ is not identified by the model-based DI estimator.}

\blue{We now use data-driven DI estimator to infer the causal connectivity from the simulated data. The pairwise DI is estimated between all pairs of these four nodes using the data-driven estimator, resulting in a $4\times 4$ matrix with zeros on the diagonal. In contrast to the model-based DI estimator, we find that DI from $\mathbf{A}$ to $\mathbf{B}$ estimated using data-driven DI is nonzero. Specifically, we find that $\hat{I}\left(\mathbf{A} \rightarrow \mathbf{B}\right)=0.433 \pm 0.011$. In addition, we also find that $\hat{I}\left(\mathbf{A}\rightarrow \mathbf{C}\right)=0.320 \pm 0.010 \: \text{and} \: \hat{I}\left(\mathbf{B}\rightarrow \mathbf{C}\right)=0.753 \pm 0.009$. To eliminate indirect causal connections, we estimated $\hat{I}\left(\mathbf{A}\rightarrow \mathbf{C} \|\mathbf{B}\right) = 0 $ and $\hat{I}\left(\mathbf{B}\rightarrow \mathbf{C} \| \mathbf{A}\right) = 0.262 \pm0.037$. Therefore, $\mathbf{A}$ to $\mathbf{C}$ is an `indirect' connection via $\mathbf{B}$. This procedure is continued to identify and remove all indirect causal connections. The resultant estimated direct causal network is depicted in Fig.~\ref{Fig:Nonlinear_Four_Node}b. It is clear that data-driven DI estimator recovers the true network correctly, while the model-based DI estimator could not infer the true causal network correctly.}

The \blue{five} diverse simulated models considered in this section demonstrate that the DI correctly infers the presence and tracks the strength of a causal connection - large values of DI imply a strong causal connection and vice versa. Using stationary bootstrap, we also showed that only large positive DI estimates correspond to statistically significant causal connections. \blue{We also observed that model-based DI estimator cannot identify nonlinear causal connections, whereas data-driven DI estimator can correctly identify both linear and nonlinear causal connections.} We now use both the MVAR model-based and data-driven DI estimators to infer the causal connectivity graph from ECoG data in epileptic patients. We only consider the large DI estimates (large compared to the rest of the causal connectivity graph) since they only imply a significant causal connection. We propose a model-based and a data-driven SOZ identification algorithm in the following section.
\section{Seizure Onset Zone Identification Algorithms}
Seizure onset zone (SOZ) is defined as the regions of the brain that initiate seizures \cite{Luders2006}. The current clinical standard is for neurologists to identify SOZ from visual analysis of the ECoG data. The SOZ identified in this way is removed during resective surgery. However, visual analysis is time consuming, subjective and potentially unreliable \cite{mierlo2013, panzica2013}. We propose two computationally derived SOZ identification algorithms - model-based and data-driven SOZ identification algorithms. We identified the SOZ in five patients with epilepsy using these two algorithms and compared their performance with visual analysis by the neurologist.

\begin{table}
\centering
\caption{Clinical Details of the Patients Analyzed.}
\begin{threeparttable}
\centering
\label{Table0}
\setlength{\tabcolsep}{1pt}
\begin{tabular}{|c|c|c|c|c|c|c|}
\hline
Patient ID & Age/Sex & \multicolumn{1}{c|}{Syndrome} & \multicolumn{1}{c|}{\makecell{Seizure \\ Type}} & \multicolumn{1}{c|}{\makecell{Electrode \\ Type}} & Surgery & \multicolumn{1}{c|}{\makecell{ Outcome \\ of Surgery}} \\
\hline
P1 & 20/M & \multicolumn{1}{c|}{\makecell{Nonlesional \\ temporal}} & CPS & D & \multicolumn{1}{c|}{\makecell{Right TL}} & Class I \\
\hline
P2 & 60/M & \multicolumn{1}{c|}{\makecell{Lesional \\ temporal}} & CPS & D & \multicolumn{1}{c|}{\makecell{Selective \\ Left HC}} & Class II \\
\hline
P3 & 29/M & \multicolumn{1}{c|}{\makecell{Nonlesional \\ temporal}} & CPS & G+D &\multicolumn{1}{c|}{\makecell{Right TL}} & Class II \\
\hline
P4 & 37/M & \multicolumn{1}{c|}{\makecell{Nonlesional \\ extratemporal}} & SPS+CPS & G &\multicolumn{1}{c|}{\makecell{Right OC}} & Class III \\
\hline
P5 & 20/F & \multicolumn{1}{c|}{\makecell{Lesional \\ temporal}} & CPS & G+D &\multicolumn{1}{c|}{\makecell{Left TL}} & Class I  \\
\hline
\end{tabular}
\begin{tablenotes}
\item CPS - complex partial seizures, SPS - simple partial seizures. D - depth electrodes, G - subdural grid electrodes. TL - temporal lobectomy, HC - hippocampectomy, OC - occipital corticosectomy. The outcomes are in Engel epilepsy surgery outcome scale. ``Class I - free of disabling seizures, class II - Almost seizure-free, class III - worthwhile improvement, class IV - no worthwhile improvement"\cite{tonini2004}.
\end{tablenotes}
\end{threeparttable}
\squeezeup
\end{table}

\subsection{Clinical ECoG Data}
The five patients analyzed here were all managed and treated by our physician coauthors. The clinical details of these patients are summarized in Table~\ref{Table0}. Three seizure records each from patients P1, P2 and P5, two from patient P3 and one from patient P4 were  analyzed. Each seizure record was approximately 10 minutes long and contained one seizure. Each seizure on average lasted for a minute and   was roughly in the middle of the seizure record. The seizure start time was identified by the neurologist. Each electrode records the voltage waveform at a sampling frequency of $1$ KHz. The number of electrodes in these five patients varied from $120$ to $150$. Electrodes with artifacts likely due to either loose contacts, patient movement or excessive line noise were not included in the analysis.

\subsection{Proposed SOZ Identification Algorithms}
The first stage of the proposed SOZ identification algorithm is an energy detector which selects only $M$ channels out of all ECoG channels for further analysis. The main objective of this stage is to reduce the computational complexity of the proposed algorithms. The energy is $l_2$-norm of the ECoG signal computed from a window around the start of seizures containing preictal and ictal recordings. Any channel  involved in seizure onset is expected to have interictal spikes before the seizure starts and/or have high amplitude low-frequency ictal activity once the seizure is fully developed, both of which will increase the energy in the selected time-window. The time-window was selected to be long enough to capture both spiking and large ECoG amplitudes during seizures. 
The second stage consisted of estimating the causal connectivity between every pair of $M$ channels selected in the first stage to form a $M \times M$ causal connectivity matrix. The causal connectivity was estimated from a shorter time-window around the seizure start time, since we are interested in estimating the seizure onset electrodes. The following subsections describe the remaining stages of the two proposed SOZ identification algorithms. 
\begin{figure}
\centering
\includegraphics[width=0.75\columnwidth]{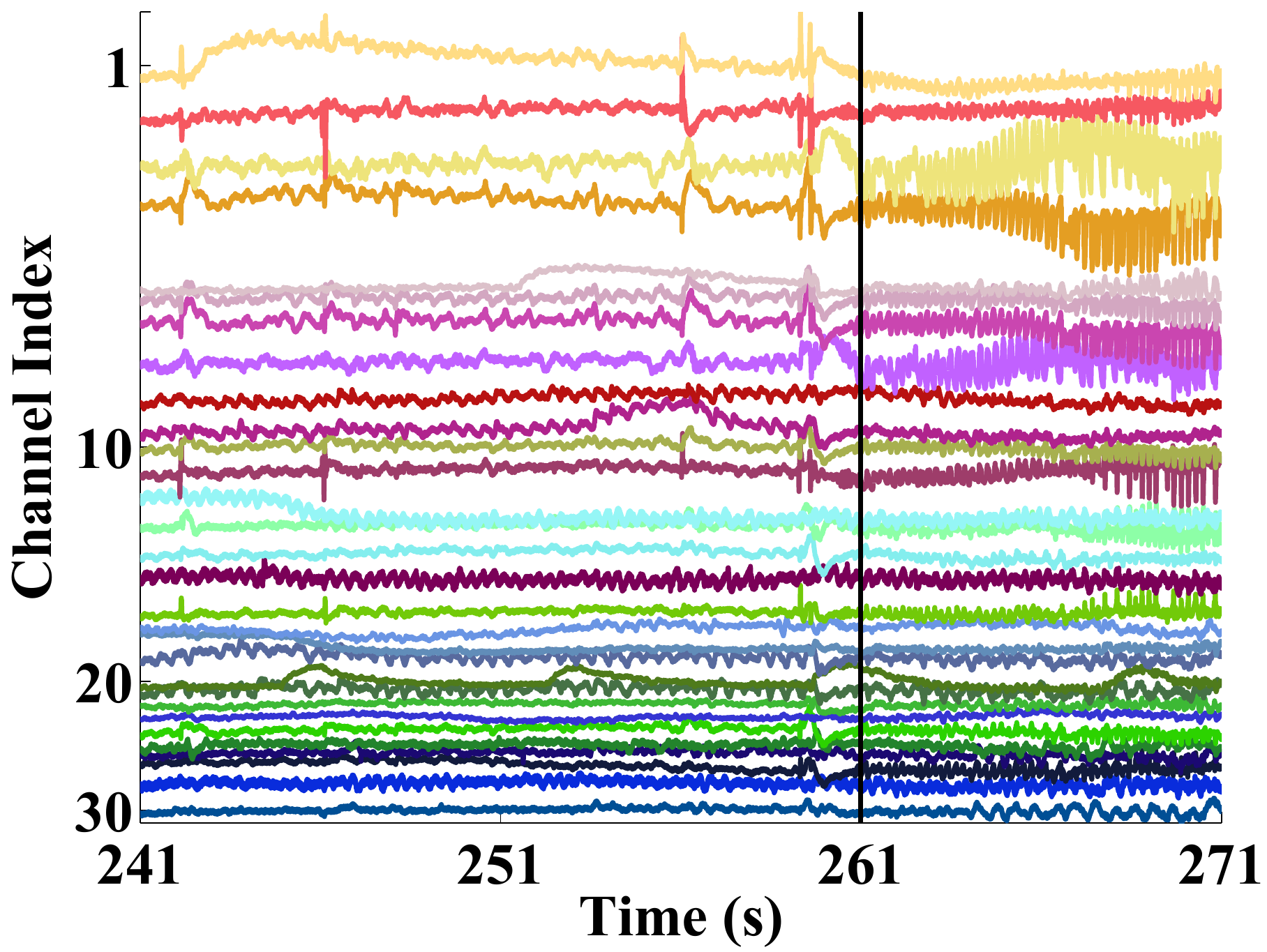}
\caption{A $30$s snapshot of ECoG signals from the $30$ high energy channels of P1. The seizure start time, represented by a vertical solid black line, is identified by neurologist. Causal connectivity is estimated from this entire $30$s window for this seizure record.}
\label{Fig:ECoG_Data}
\squeezeup
\end{figure}

\subsubsection{Model-based SOZ Identification Algorithm}
In this approach, ECoG data is assumed to be derived from a MVAR process with Gaussian white noise. This is a very common assumption imposed to estimate causal connectivity between ECoG data \cite{mierlo2014, Blinowska2011}. The MVAR model-based DI estimator is used to infer the causal connectivity between the selected $M$ high energy channels. The causal connectivity estimated using this approach only represents the linear causal interactions between the ECoG channels. However it is widely believed that seizures are highly non-linear  phenomenon during which SOZ drives the rest of the network into a hypersynchronous state \cite{lehnertz2008, Luders2006,rosenow2001}. As a result, we expect the seizure onset electrodes in the causal connectivity graph to be isolated, since model-based approach can only capture linear causal interactions. The proposed model-based algorithm therefore identifies the nodes in the causal connectivity graph with zero degree (threshold was set to select only the strongest $10\%$ connections) as the estimated SOZ. If a patient had multiple seizures, the electrodes identified across all seizures in that patient form the estimated SOZ for that patient.

\begin{figure}
\subfloat[From model-based algorithm]{
\centering
\includegraphics[width=0.48\columnwidth]{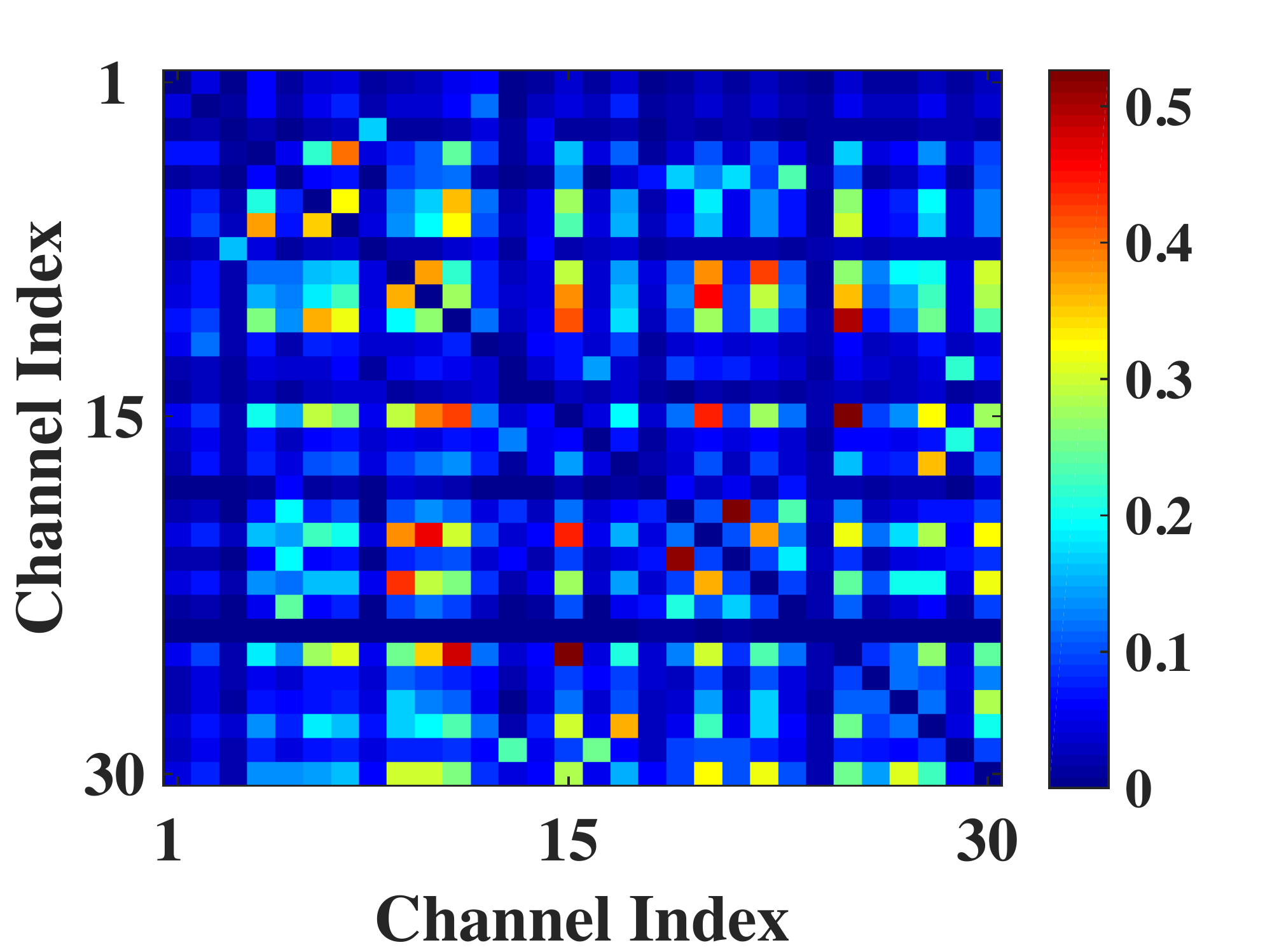}}
\subfloat[From data-driven algorithm]{
\centering
\includegraphics[width=0.48\columnwidth]{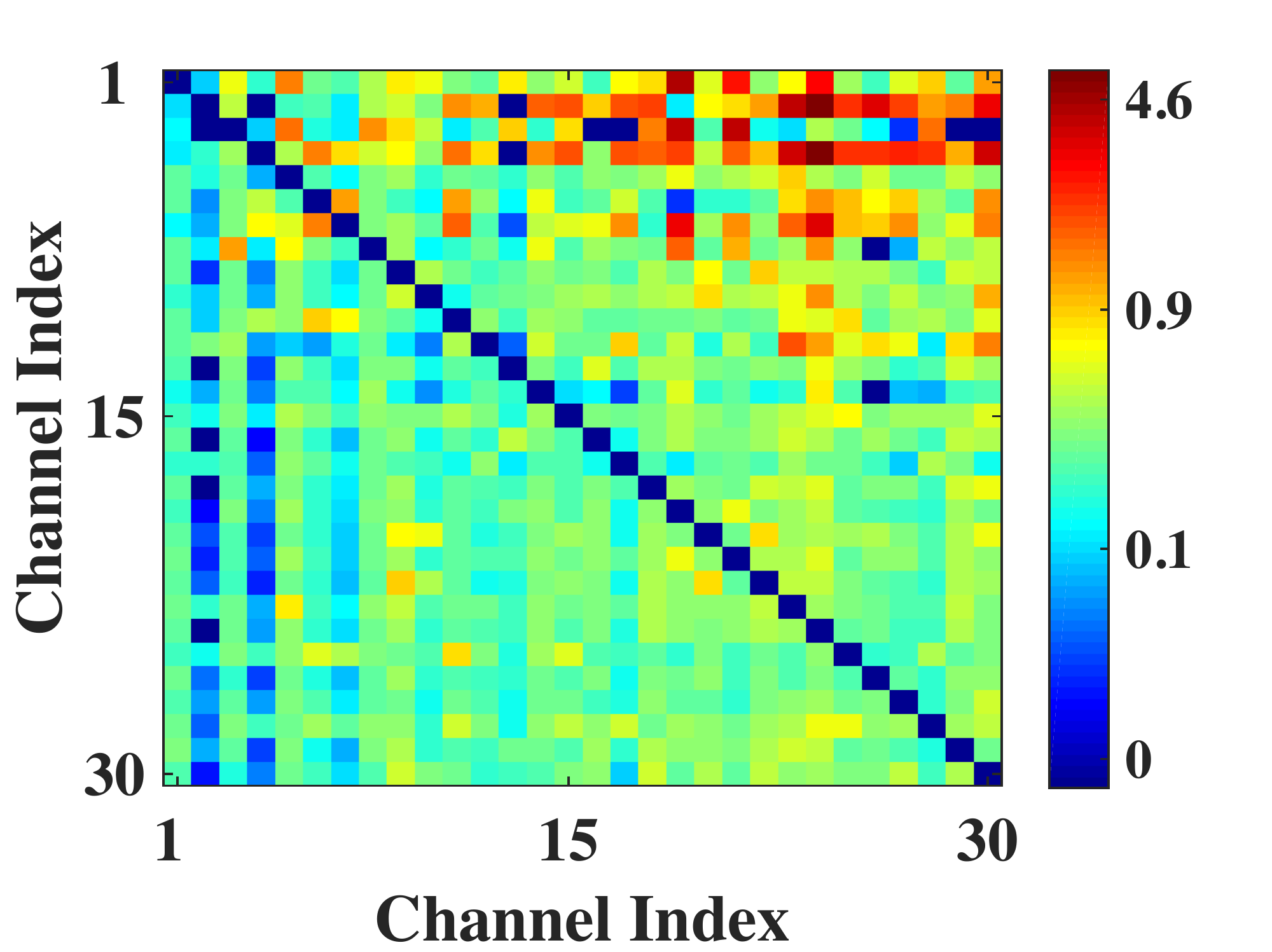}}
\caption{Causal connectivity between $30$ high energy channels estimated from ECoG data between $241$s and $271$s from the second seizure of P1. The channel indices with bluish rows and bluish columns (correspond to low DI estimates) in Fig.~\ref{Fig:P1Causal_connectivity}a correspond to isolated nodes and are the estimated SOZ using model-based algorithm. The corresponding channels in Fig.~\ref{Fig:P1Causal_connectivity}b have large net-outflows of information and are the estimated SOZ from data-driven algorithm.}
\label{Fig:P1Causal_connectivity}
\squeezeup
\end{figure}
 
\subsubsection{Data-driven SOZ Identification Algorithm}
In this algorithm, no parametric model assumptions were imposed on ECoG data. The causal connectivity between the $M$ high energy channels selected in the first stage was inferred using the data-driven DI estimator. This estimator inferred both linear and nonlinear causal interactions between channels. Intuitively, activity at the SOZ electrodes drives the activity at the other electrodes into a hypersynchronous state via linear and nonlinear causal interactions \cite{lehnertz2008}. We therefore expect the SOZ electrodes to act as sources (with strong outgoing and weak incoming causal connections) in the causal connectivity graph inferred around the seizure start time using data-driven DI. As a result, the SOZ nodes in the causal connectivity graph are expected to have large net-outward flow of information. The data-driven SOZ identification algorithm quantifies this intuition to estimate SOZ. The net-outward flow ($\Phi$) of causal information from an electrode $i$ is calculated using 
\begin{align} \label{nof}
\Phi(i) = \sum\limits_{j=1,j \neq i}^{M} \left\{I(i \rightarrow j) - I(j \rightarrow i)\right\}.
\end{align}
If a patient had multiple seizures, the net outward flow of an electrode is the average net outward flow of that electrode across all seizures recorded in that patient. Then the normalized net outward flow ($\tilde{\Phi}$) of the electrode $i$ is given by 
\begin{align} \label{pnof}
\tilde{\Phi}(i) = 100 \times \frac{\Phi(i)}{\sum\limits_{j:\Phi(j)>0} \Phi(j)}.
\end{align}
The electrodes with $\tilde{\Phi} > 5\%$ are considered to have significant net outward flow of information in the causal connectivity graph and are identified as the seizure onset electrodes for that patient by the data-driven SOZ identification algorithm. 

\subsection{Performance of Proposed SOZ Identification Algorithms}
The energy detector selected the top $M=30$ channels with the largest energy computed from a $100$s window comprising of $50$s of activity immediately before and after the seizure starts. The causal connectivity graph between these high energy channels is then estimated using model-based and data-driven DI estimators from a $30$s window that begins $20$s before the seizure start time \blue{and ends $10s$ into the start of the seizure}. We assumed that the current activity at an ECoG channel does not depend on more than $150$ms of past activity ($150$ past samples at $F_s=1$KHz) at this channel and other channels. This corresponds to restricting the model order $J_{yy},K_{yx}$ search space to $\left[1,150\right]$ for the MVAR model-based DI estimator. 
In addition, we need to capture the connectivity just before and just after a seizure starts to estimate the SOZ. Therefore, we used ECoG data from a $30s$ window ($3\times 10^4$ data points) that begins $20s$ before the start of the seizure to be stationary. The same window was used for the data-driven estimator as well. In addition, the past activity was down-sampled by a factor of $50$ for the data-driven estimator to restrict the $J_{yy},K_{yx}$ search space to $\left[1,4\right]$ and also reduce its computational complexity (i.e. the past activity of channel $\mathbf{X}$ can include $\left\{x_n,x_{n-50},x_{n-100},x_{n-150}\right\}$). The exact values of these parameters is not crucial as the algorithms seem to be fairly robust to changes in these parameters.

Consider the second seizure record of patient P1. The energy detector selected $30$ high energy channels. Fig.~\ref{Fig:ECoG_Data} shows the recordings from these channels in the $30$s window in which causal connectivity graph is inferred. The inferred graph by model-based and data-driven approaches is shown in Fig.~\ref{Fig:P1Causal_connectivity}. The weighted adjacency matrix of the inferred causal connectivity graphs, whose $(i,j)^{th}$ element is the DI estimate from channel $i$ to $j$ for $i,j\in \left[1,30\right]$, is plotted in Fig.~\ref{Fig:P1Causal_connectivity} using a image plot. It is clear from this figure that the mean strength of the DI estimates using model-based approach is smaller than using data-driven approach (colorbar ranges are different in the two sub-figures). We observed this across all the twelve seizures analyzed. This indicates  that data-driven DI captured more causal information on average than model-based DI, implying that non-linear causal interactions are stronger around the beginning of a seizure. \blue{The nodes with zero degree in the causal connectivity graphs from each seizure in a patient are identified as the SOZ by the model-based algorithm. The zero degree criterion used by model-based algorithm is counterintuitive, since we except the SOZ to drive the network to seizure state and not be weakly connected. On the other hand, the data-driven algorithm selects electrodes with large net outflows, which is very intuitive.}  The data-driven algorithm computed the normalized net outward flow for each node using \eqref{pnof}. Fig.~\ref{Fig:PNOF}a plots the $\tilde{\Phi}$ for all electrodes with positive net outward flows in patient P1. The electrodes with $\tilde{\Phi} > 5\%$ are the estimated SOZ for this patient P1 using data-driven algorithm.
\begin{figure}
\subfloat[From Patient P1]{
\centering
\includegraphics[width=0.48\columnwidth]{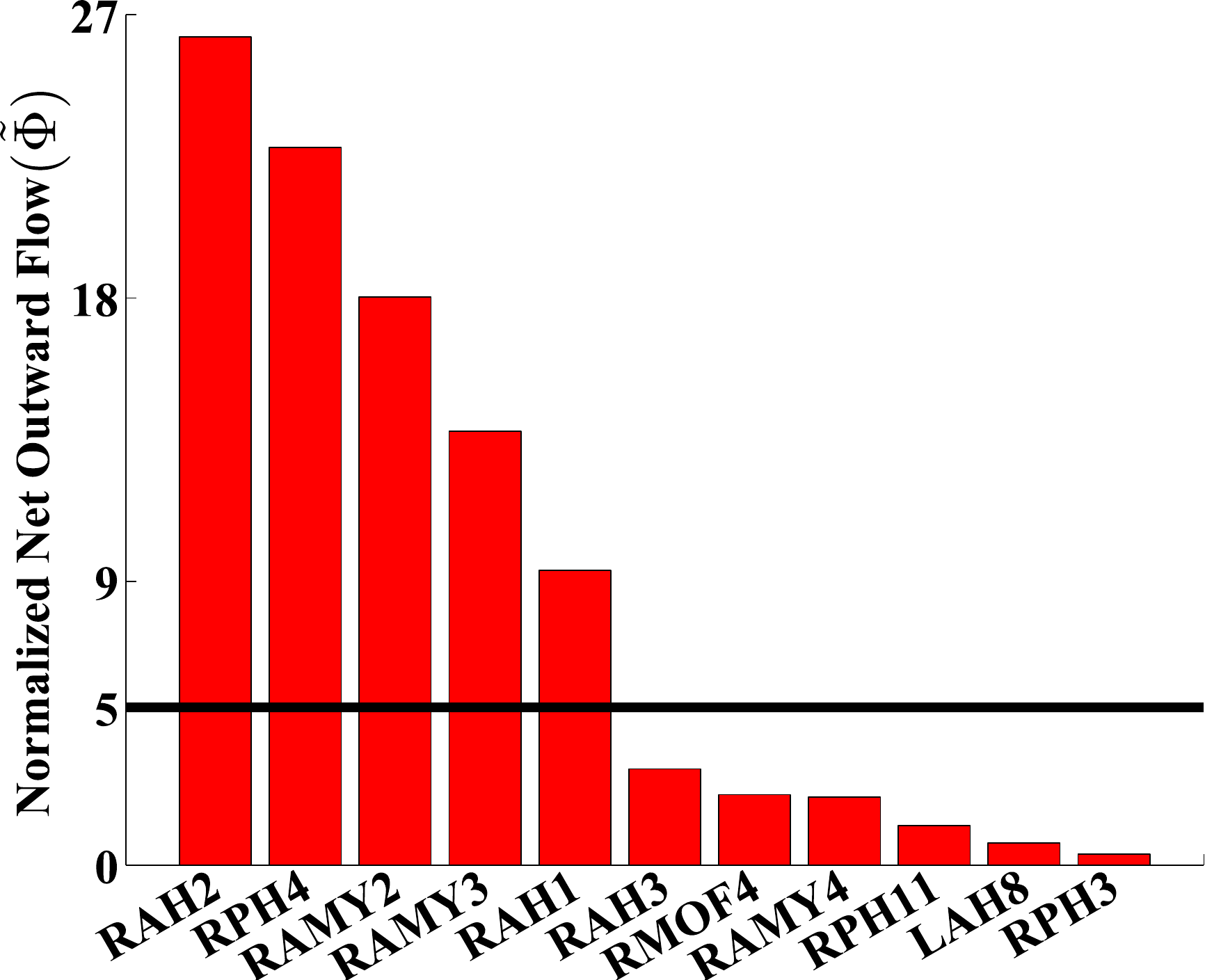}}
\subfloat[From Patient P3]{
\centering
\includegraphics[width=0.48\columnwidth]{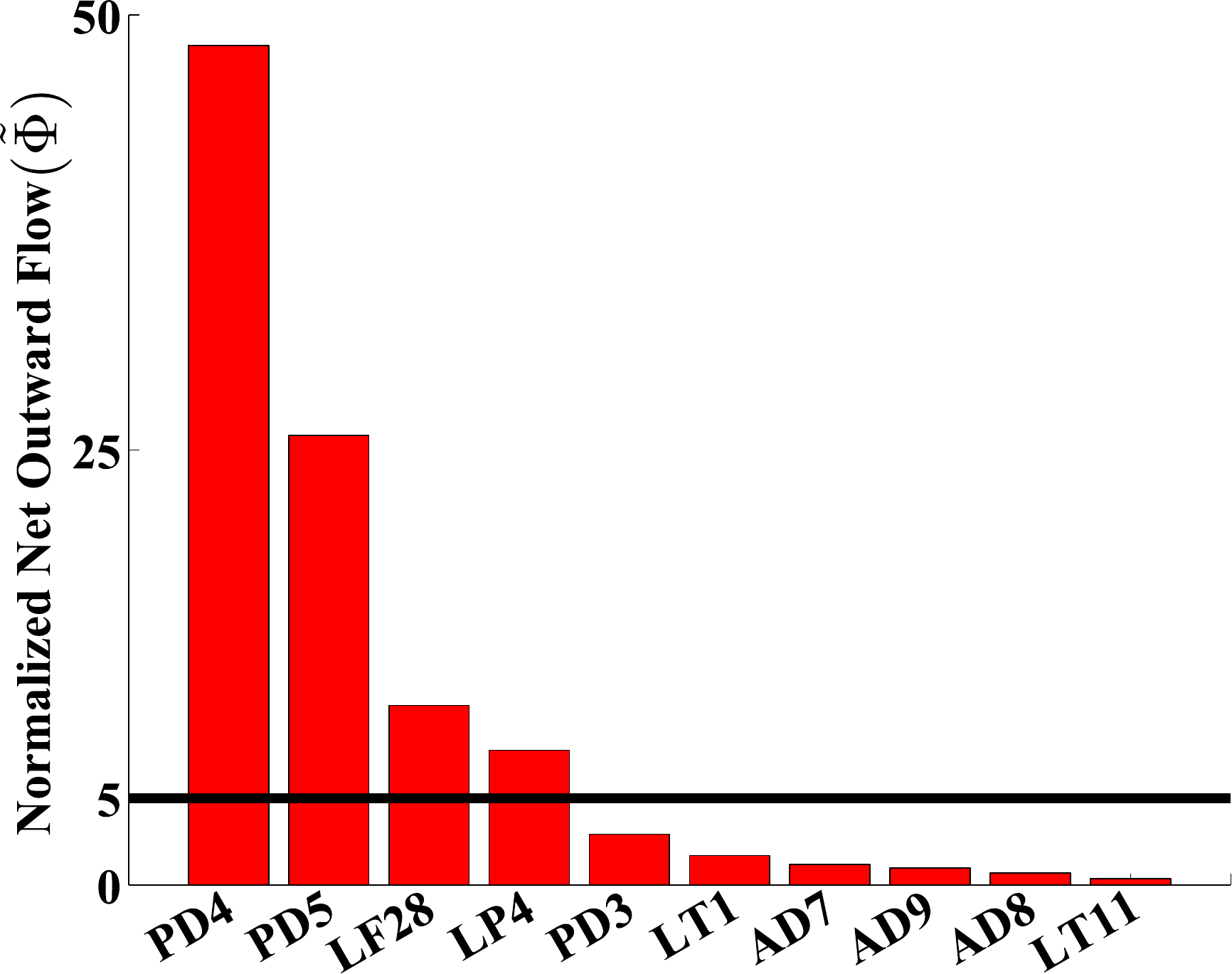}}
\caption{Normalized net outward flow from the ECoG electrodes with positive net information outflow using data-driven SOZ identification algorithm.} \label{Fig:PNOF}
\squeezeup
\end{figure}

Table~\ref{Table1} summarizes the results from our analysis. The first column in Table~\ref{Table1} identifies the patient ID and the number of seizures analyzed for that patient. The second, third and fourth columns in Table~\ref{Table1} list the SOZ identified across all the five patients using model-based, data-driven algorithms and visual analysis respectively. \blue{We observed that all the channels identified as SOZ by visual analysis, except AST 2 in one seizure of P5, are included in the $30$ high energy channels selected from each seizure by the energy detector in the first stage. The top $30$ channels selected from two seizures in patient P5 contained AST 2, but the $30$ channels picked from the third seizure did not contain AST 2. The normalized net outflow $\tilde{\Phi}$ from AST 2 electrode for patient P5 using data-driven algorithm was $1\%$ and hence this electrode was not identified as SOZ (note that $\tilde{\Phi}$ has to exceed $5\%$ to be selected as SOZ). Expect for this one region,}  it is clear from this table that the  data-driven algorithm identifies all the regions identified by the neurologist, whereas the model-based algorithm misses some regions (for instance, RAMY electrodes in P1, TP and AST electrodes in P5). Also, the model-based algorithm incorrectly identified lateral temporal (LT) electrodes as SOZ in patient P3, whereas data-driven algorithm correctly identified posterior depth (PD) electrodes in hippocampus as SOZ. Except in P3 and P4, both algorithms do not have any false positives. The false positives in P4 could be because only one seizure was analyzed in this patient.

Another advantage of the data-driven SOZ identification algorithm over model-based algorithm and analysis by the neurologist is that $\tilde{\Phi}$ could be used as a quantitative metric to rank the electrodes in the decreasing order of clinical relevance. Fig.~\ref{Fig:PNOF} plots the $\tilde{\Phi}$ of all electrodes with positive net outward flows in patients P1 and P3. Using our metric $\tilde{\Phi}$, it is clear from Fig.~\ref{Fig:PNOF}b that electrodes PD4, PD5 contribute much more in generating and spreading seizures than LF28 and LP4 electrodes even though $\tilde{\Phi}$ exceeds the chosen threshold at all these four electrodes. Depending on the significance level ($5\%$ is used here), the set of selected SOZ electrodes varies. We observed in all five patients that the electrodes with the highest $\tilde{\Phi}$ values were always the same as the ones identified by the neurologist. \blue{Visual analysis by the neurologist can only give qualitative information about the SOZ and cannot give quantitative information like the proposed data-driven SOZ identification algorithm.} In addition, data-driven algorithm can also differentiate between electrodes in close proximity - for example, $\tilde{\Phi}$ is negative for RPH2 electrode in P1 even though $\tilde{\Phi}$ is positive for both RPH3 and RPH4 (refer to Fig.~\ref{Fig:PNOF}a). The increased spatial-specificity provided by our data-driven algorithm could be relevant for next generation epilepsy treatments  \cite{krook2015}. The main advantage of the model-based algorithm over data-driven one is its lower computational complexity. However, this is less critical with today's powerful computers. To summarize, data-driven SOZ identification algorithm outperforms model-based algorithm and provides more interpretable results.
\begin{table}
\centering
\caption{Seizure onset zone identified from the proposed algorithms and the visual analysis by neurologist.}
\begin{threeparttable}
\label{Table1}
\setlength{\tabcolsep}{1.5pt}
\begin{tabular}{|c|c|c|c|}
\hline
\multicolumn{1}{|c|}{\makecell{\textbf{Patient - \#}\\ \textbf{of Seizures}}} & \textbf{\makecell{Model-based \\ Algorithm}} & \textbf{\makecell{Data-driven \\ Algorithm}} & \multicolumn{1}{c|}{\textbf{Visual Analysis}} \\
\hline
\multicolumn{1}{|c|}{P1 - 3} & {{\textit{RAH 1-3}, \textit{RPH 2-4}}} & \multicolumn{1}{c|}{\makecell{\textit{RAH 1-2}, \textit{RPH 4}, \\ RAMY 2-3}} & \multicolumn{1}{c|}{\makecell{RAH 1-3, RPH 2-4, \\ RAMY 2-3}}\\
\hline
{P2 - 3} & {{\textit{LAH 2-4}, \textit{LPH 1-2}}} & \textit{LAH 2-4, LPH 2} & LAH 2-4, LPH 1-2\\
\hline
{P3 - 2} & {LT 1-3, 10} & \textit{PD 4-5}, LF 28, LP 4 & {PD 3-5}\\
\hline
\multicolumn{1}{|c|}{P4 - 1} & {\makecell{\textit{LO 3, 14, 15, 25}, \\ LO 12, 13, \textit{PST 3}, \\ PST 1, MOG 27}} & {\makecell{\textit{LO 3, 14, 15}, 12, \\ PST 1, MOG 23, \\ SOG 21, 36 }} & \multicolumn{1}{c|}{\makecell{LO 3, 14, 15, \\ LO 25, PST 3}}\\
\hline
\multicolumn{1}{|c|}{P5 - 3} & {{\textit{MST 1, 2}, \textit{HD 1}}} & \multicolumn{1}{c|}{\makecell{\textit{MST 1, TP 1}, \\ \textit{HD 1} }} & \multicolumn{1}{c|}{\makecell{MST 1, 2, TP 1,\\ HD 1-3,  AST 2}}\\
\hline
\end{tabular}
\begin{tablenotes}
\item The label of an ECoG electrode comprises of an abbreviation of the brain region it is implanted in and a number. For depth electrodes, smallest number is assigned to deepest electrode from scalp. For instance, RAH1 - deepest electrode contact in depth electrode in right anterior hippocampus and LO3 - third electrode contact in subdural grid electrode over lateral occipital lobe. RPH - right posterior hippocampus, RAMY - right amygdala, LF - lateral frontal, LP - lateral parietal, LT - lateral temporal, PD - posterior hippocampal depth, MOG - medial occipital grid, SOG - sub-occipital grid, PST - posterior sub-temporal, MST - mid-subtemporal lobe AST - anterior sub-temporal lobe, TP - temporo-polar, HD - hippocampal depth.
\end{tablenotes}
\end{threeparttable}
\squeezeup
\end{table}

\section{Discussion and Conclusions}
\begin{figure*}[!t]
\centering
\begin{minipage}[b]{0.24\textwidth}
\centering
\includegraphics[width=\textwidth]{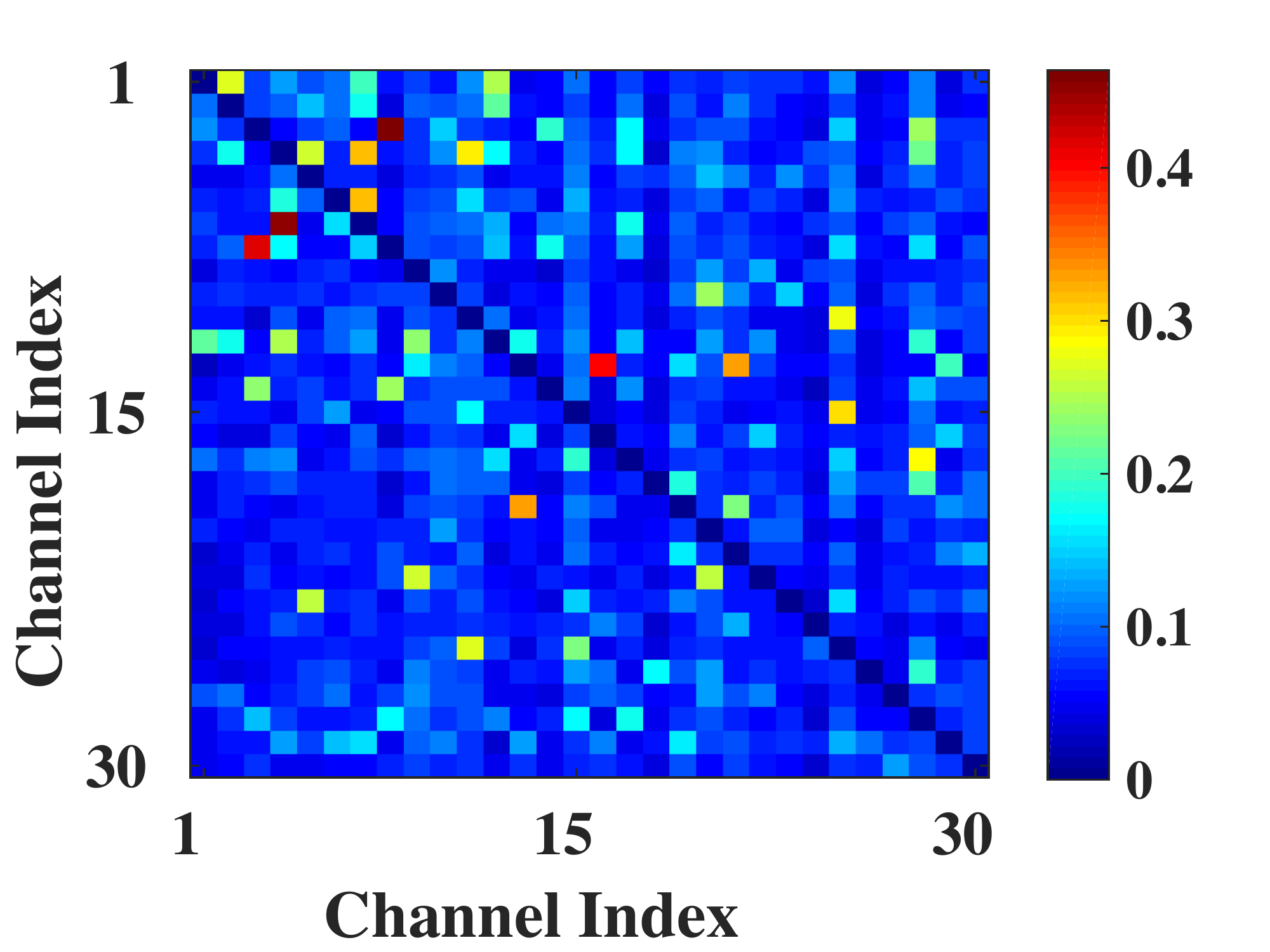}
\caption{\blue{Causal connectivity between 30 high energy channels depicted in Fig.~\ref{Fig:ECoG_Data} estimated using PDC.}}
\label{Fig:PDC}
\end{minipage}
\vrule 
\begin{minipage}[b]{0.75\textwidth}
\centering
\subfloat[From a segment before the seizure]{
\includegraphics[width=0.3\textwidth]{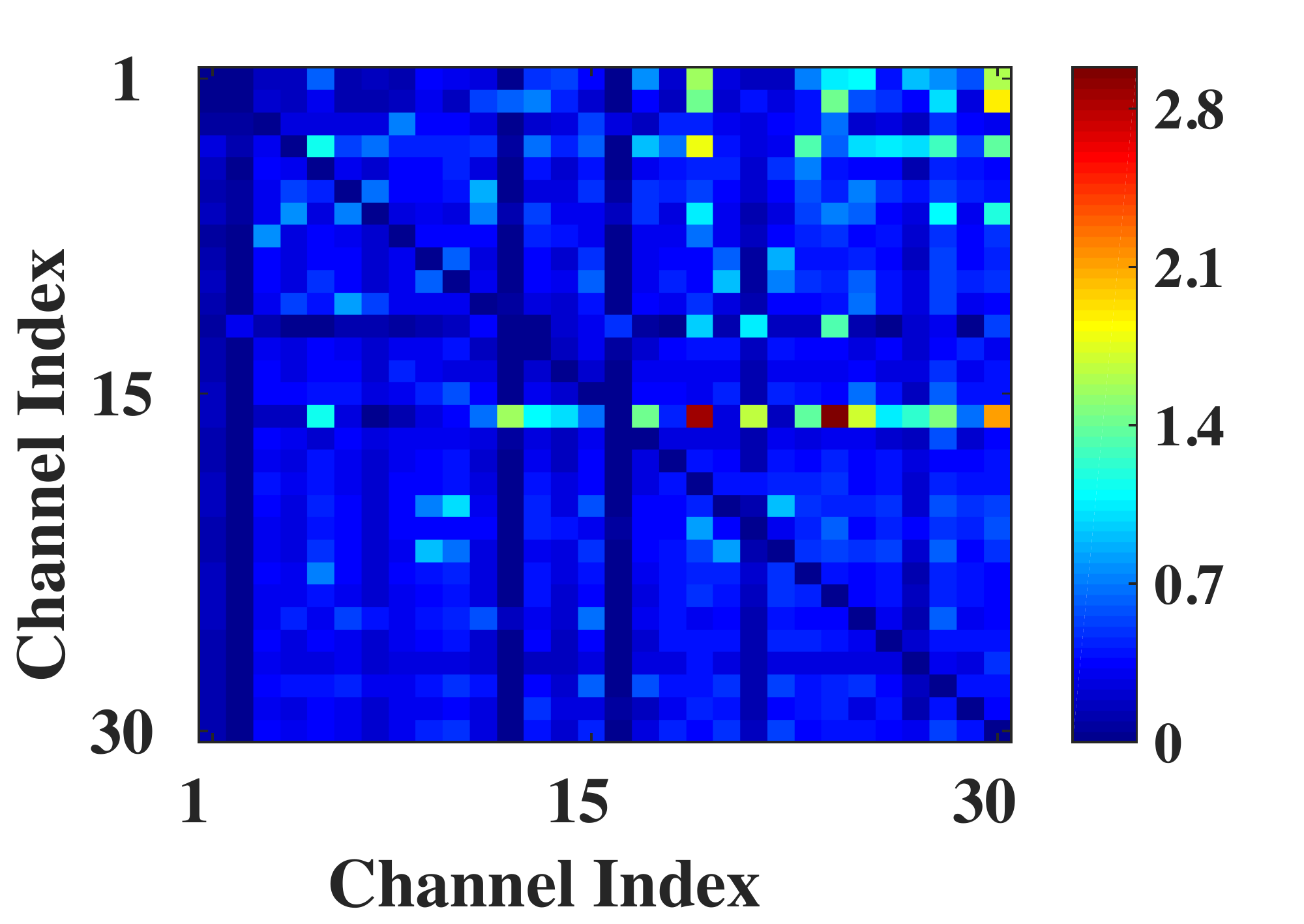} 
\label{Fig:Before_Seizure}}
\hfill
\subfloat[From a segment during the seizure]{
\includegraphics[width=0.3\textwidth]{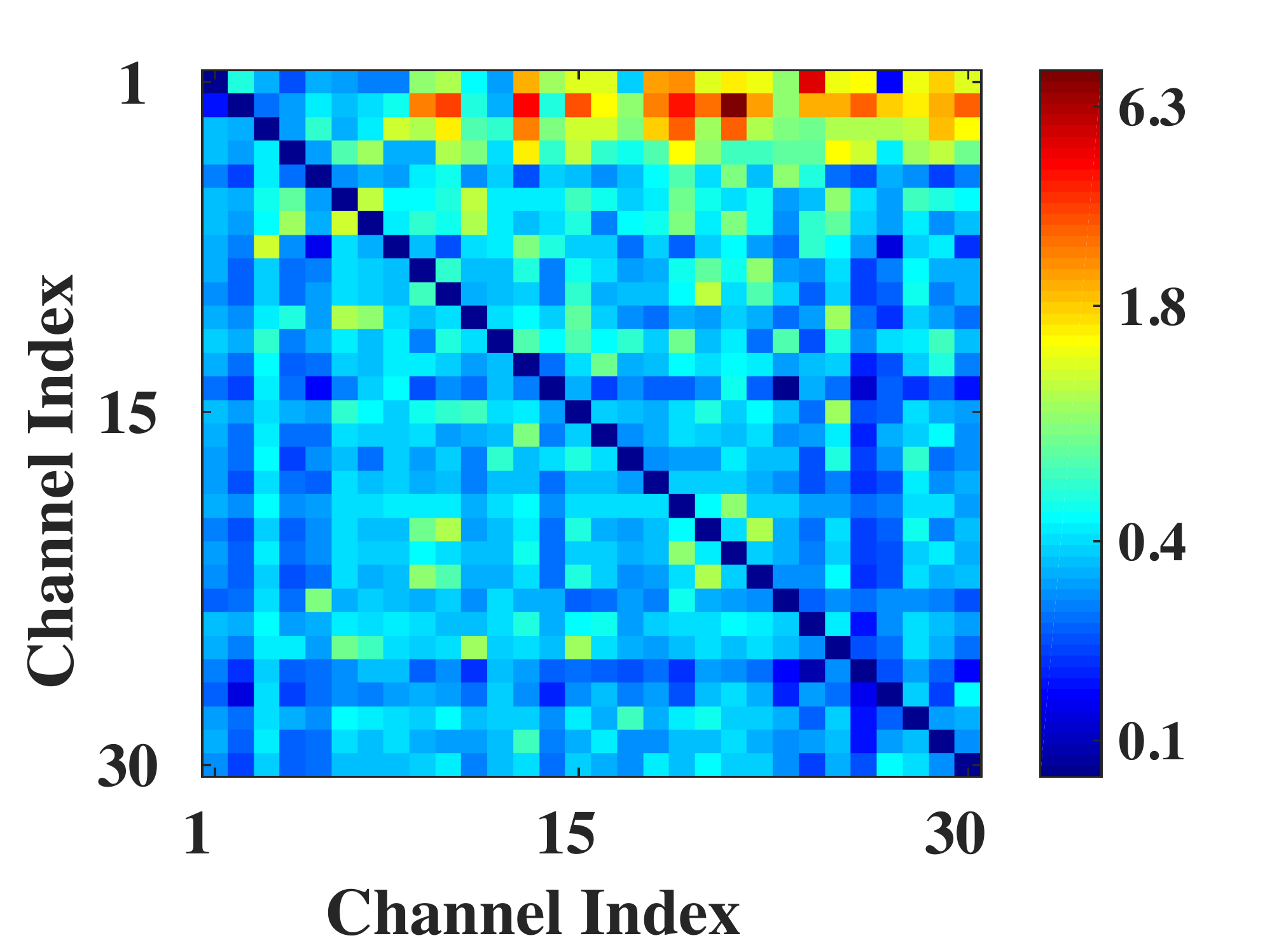} 
\label{Fig:During_Seizure}}
\hfill
\subfloat[From a segment after the seizure]{
\includegraphics[width=0.3\textwidth]{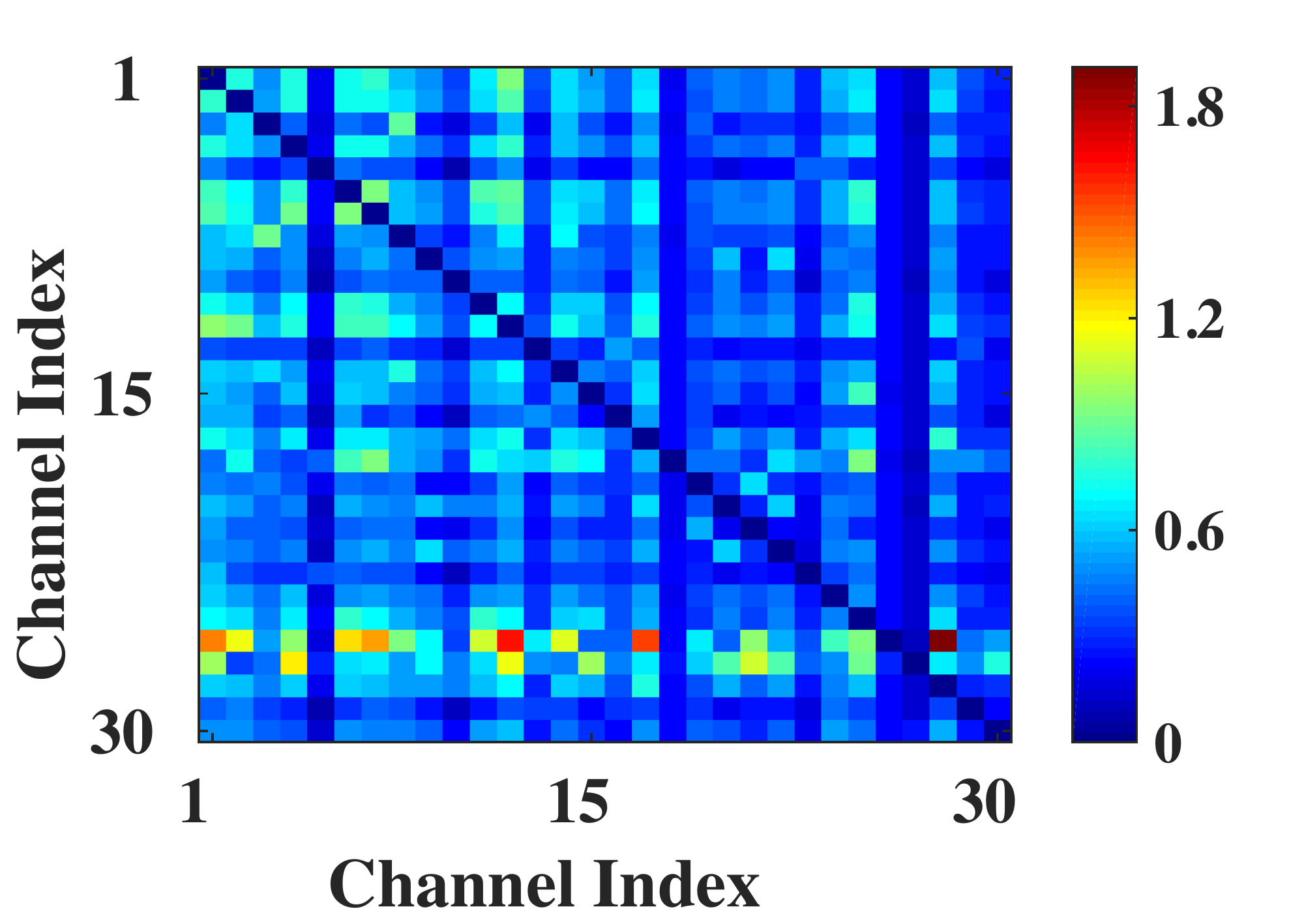} 
\label{Fig:After_Seizure}}
\caption{\blue{Causal connectivity between the 30 high energy channels from the second seizure of patient P1 estimated using data-driven DI estimator from ECoG data in three segments - one before seizure (181s -211s), one during seizure (261s - 291s) and one after seizure (361s - 391s). This seizure starts at 261s and ends at 350s. 
}}\label{Fig:Dynamic_Connectivity}
\end{minipage}
\squeezeup
\end{figure*} 

An almost surely convergent MVAR model-based and data-driven estimators for DI are introduced in this paper. Linear causal interactions between two time-series can be quantified using MVAR model-based DI estimator, whereas both linear and nonlinear causal interactions are quantified by data-driven DI estimator. The resultant DI estimates can be used to infer whether the data has (1) linear causal interactions or (2) both linear and nonlinear causal interactions. If the MVAR model-based DI estimate is comparable in value to data-driven DI estimate, then the interaction is predominantly linear. This is not feasible with existing metrics because they can be split into two non-overlapping groups - the first group only detects linear causal interactions (e.g., Granger causality, partial directed coherence), whereas the second group detects both linear and nonlinear causal interactions (e.g., transfer entropy). The DI estimators proposed in this paper can be automatically adapted to other types of electrophysiological data like EEG to learn the causal connectivity.

Data-driven DI estimator seems to be more appropriate than model-based DI estimator if the underlying data distribution is not known, which is the case with most real data. The main challenge with data-driven DI estimator is estimating the causal conditional likelihood nonparametrically and its computational complexity. We used kernel density estimators in this paper to estimate causal conditional likelihood. Kernel density estimators are asymptotically optimal \cite{scott2015}. Their bias decreases with increasing number of data samples and complexity increases with the dimensionality of the data, just like other nonparametric estimators. Even though we selected optimal bandwidth using smoothed cross-validation to minimize the asymptotic mean integrated squared error, several other criteria could also be used \cite{scott2015,duong2007}. In addition, data-driven entropy estimators based on adaptive partitioning, nearest neighbors and m-spacing algorithms \cite{wang2009,Liu2012} can also be used to estimate DI nonparametrically. Another approach to estimate DI nonparametrically is to extend the universal DI estimator proposed for discrete-valued signals in \cite{Weissman2013a} to continuous-valued ECoG signals. Future work should also include developing approximate data-driven DI estimators to further reduce computational complexity.

Directed information was used in this paper to estimate causal connectivity between ECoG channels. The causal connection identified between two channels could be due to the effect of activity at other spatial locations in the brain. If an ECoG electrode was implanted at these other locations, causally conditioned DI can be used to remove their influence. This was demonstrated using the \blue{four} node examples in section~\ref{sec:DI_Simulated_Data}. On the other hand, if ECoG activity is not recorded from these locations, then removing the effects of these hidden nodes on the inferred causal connectivity is a very hard problem in general. Future work should look into the sensitivity of DI to volume conduction effects when compared with synchronization metrics like phase lag index \cite{stam2007}. 

\blue{DI estimators proposed in this paper do not quantify the amount of causal information between time-series at each frequency, unlike partial directed coherence (PDC) or directed transfer function (DTF). However, the advantage of DI is that data-driven DI estimator can detect nonlinear causal interactions, which PDC or DTF cannot detect. Metrics based on PDC, DTF assume the data is drawn from a MVAR model and can only detect linear causal interactions (similar to MVAR model-based DI estimator proposed in this paper). To demonstrate this, we estimated the causal connectivity graph between the 30 channels depicted in Fig.~\ref{Fig:ECoG_Data} by PDC using eMVAR toolbox\cite{faes2011}. The resultant $30 \times 30$ causal connectivity matrix is plotted in Fig.~\ref{Fig:PDC}, in which $\left(i,j\right)$ element corresponds to the maximum value of PDC from channel $i$ to channel $j$. Note that causal connectivity estimates from the proposed DI estimators for the same data is plotted in Fig.~\ref{Fig:P1Causal_connectivity}. 
Comparing Fig.~\ref{Fig:PDC} with Fig.~\ref{Fig:P1Causal_connectivity}b, it is clear that net outflow from the SOZ electrodes is not large in PDC when compared to data-driven DI. This implies unlike data-driven DI estimator, PDC cannot capture nonlinear causal interactions.}

We also proposed model-based and data-driven algorithms to identify the SOZ. The first stage of both these algorithms is an energy detector. The chosen electrodes from the first stage turned out to have large overlap (more than half) across multiple seizures within a patient. All electrodes with low rhythmic gamma activity in SOZ were selected by the energy detector in all the patients analyzed. Note that other criteria could also be used instead of energy detector. In particular, we experimented with selecting channels displaying strong high-frequency activity around the seizure start time (since channels involved in seizure onset display strong high-frequency activity around the beginning of a seizure that typically develops into high amplitude low-frequency activity). The time-window used to estimate the high-frequency activity should be of much smaller length than the one used with energy detector, because the seizures typically display low amplitude rhythmic high-frequency oscillations only for a very short duration. The resulting performance with energy detector or the high-frequency activity detector was similar. We therefore presented the results only with the energy detector in this paper.

The causal connectivity graphs between the selected high energy channels estimated using MVAR model-based DI and data-driven DI from same time-window are not the same, since both estimators capture different causal interactions in the data - model-based captures linear interactions, whereas data-driven captures both linear and nonlinear causal interactions. Therefore the criterion used to estimate SOZ from the causal connectivity graph was different for the two algorithms. In the model-based approach, the SOZ nodes are isolated since they drive the other brain regions into a seizure through nonlinear interactions (which are not captured by model-based DI estimator). Similar results were reported in other studies using linear metrics \cite{warren2010,burns2014}. It is reported in \cite{burns2014} that SOZ electrodes form an isolated focus using symmetric coherence metric that captures linear interactions. On the other hand in causal connectivity graphs estimated by data-driven DI, the outgoing and incoming edges from SOZ electrodes have large and small DI estimates respectively (refer to Fig.~\ref{Fig:PNOF}). This is in accordance with our intuition that the SOZ drives the seizure activity \cite{lehnertz2008, Luders2006,rosenow2001}. Also metrics closely related to net outward flow were used in \cite{sabesan2009} to infer SOZ using transfer entropy (which detects nonlinear interactions) by analyzing hours of ECoG recordings (here we are only using recordings from a $30$s window). 

\blue{A key advantage of ECoG recordings over other neuroimaging techniques is its good temporal resolution. The DI estimators proposed in this paper can be applied to ECoG recordings from different windows to learn the spatiotemporal changes in causal connectivity networks during the course of a seizure. The causal connectivity before, during and after the second seizure of patient P1 estimated using data-driven DI estimator from three $30$s long windows is shown  in Fig.~\ref{Fig:Dynamic_Connectivity}. It is clear from Fig.~\ref{Fig:Dynamic_Connectivity} that SOZ electrodes (corresponding to rows with more red color or large DI values in Fig.~\ref{Fig:Dynamic_Connectivity}b) have large net outflows during seizure when compared with before and after seizure (same rows have more blue color or smaller DI value in Fig.~\ref{Fig:Dynamic_Connectivity}a ,~\ref{Fig:Dynamic_Connectivity}c). We are extending this analysis to infer seizure mechanisms by examining the changes in causal connectivity estimated from preictal, ictal, postictal periods when compared with interictal periods. This is the subject of our current and future work \cite{malladi2016}.} The results from this analysis potentially could improve our understanding of seizure  mechanisms and  lead to the development of novel non-surgical treatments for epilepsy.

\section{Acknowledgments}
The authors wish to thank Sugnaya Karunakaran for carefully reading the manuscript.

\begin{appendices}
\section{Proof of causal conditional entropy estimator}\label{AppendixA}
\subsection{Proof of Lemma~\ref{lemma1}}
First, we will prove the existence of $h\left(\mathbf{Y\|\mathbf{X}}\right)$. Since conditioning reduces differential entropy, we have
\begin{align} \label{AppendixA_eq0}
h\left(y_1\right) & \geq h\left(y_1|\mathbf{X}_1^1\right) \geq h\left(y_2|\mathbf{Y}_1^1,\mathbf{X}_1^2 \right) \nonumber \\
& \geq \cdots \geq h\left(y_{n}|\mathbf{Y}_{n-J}^{n-1},\mathbf{X}_{n-K+1}^{n}\right) \geq \cdots\cdots
\end{align}
Therefore the sequence $h\left(y_{n}|\mathbf{Y}_{n-J}^{n-1},\mathbf{X}_{n-K+1}^{n}\right)$ is a non-increasing sequence that is upper bounded by $h\left(y_1\right)$.
Also let $l = \max \left(J+1,K\right)$. Then for $n \geq l$,
\begin{align}
h\left(y_n|\mathbf{Y}_{1}^{n-1},\mathbf{X}_1^{n}\right) & = h\left(y_n | \mathbf{Y}_{n-J}^{n-1}, \mathbf{X}_{n-K+1}^{n}\right) \label{AppendixA_eq1}\\
& = h\left(y_l | \mathbf{Y}_{l-J}^{l-1}, \mathbf{X}_{l-K+1}^{l}\right), \label{AppendixA_eq2}
\end{align}
where \eqref{AppendixA_eq1} is from the Markovian assumption and \eqref{AppendixA_eq2} is from the stationarity assumption. Note that \eqref{AppendixA_eq2} also implies the sequence $h\left(y_{n}|\mathbf{Y}_{n-J}^{n-1},\mathbf{X}_{n-K+1}^{n}\right)$ is lower bounded by $h\left(y_l | \mathbf{Y}_{l-J}^{l-1}, \mathbf{X}_{l-K+1}^{l}\right)$. Let $a_n = h\left(y_n|\mathbf{Y}_{1}^{n-1},\mathbf{X}_1^{n}\right)$ and $b_N  =  \tfrac{1}{N} h\left(\mathbf{Y}^N \| \mathbf{X}^N \right) = \tfrac{1}{N}\sum\limits_{n=1}^{N} a_n$. Since the $\lim\limits_{N \rightarrow \infty} a_N$ exists, from Cesaro mean theorem \cite{Cover2006} we have $h\left(\mathbf{Y}\|\mathbf{X}\right) = \lim\limits_{N \rightarrow \infty} b_N$ also exists. The above proof can be easily modified to prove $h\left(\mathbf{Y}\right)$ exists. Therefore $I\left(\mathbf{X} \rightarrow \mathbf{Y}\right) = h\left(\mathbf{Y}\right) - h\left(\mathbf{Y}\|\mathbf{X}\right)$ also exists.

\subsection{Proof of Lemma~\ref{lemma2}}
\vspace*{-0.5cm}
\begin{align} 
\tfrac{1}{N}h\left(\mathbf{Y}^{N}\|\mathbf{X}^{N}\right) &  = \tfrac{1}{N} \textstyle \sum\limits_{n=1}^{N} h\left(y_n|\mathbf{Y}_{n-J}^{n-1},\mathbf{X}_{n-K+1}^{n} \right) \label{AppendixA_eq5} \\
& \hspace*{-1cm} = \tfrac{1}{N} \textstyle \sum\limits_{n=1}^{N} \mathbb{E} \left[-\log \mathrm{P}\left(y_l|\mathbf{Y}_{l-J}^{l-1},\mathbf{X}_{l-K+1}^{l} \right) \right] \label{AppendixA_eq7} \\
& \hspace*{-1cm} = \mathbb{E}\left[-\log \mathrm{P}\left(y_l|\mathbf{Y}_{l-J}^{l-1},\mathbf{X}_{l-K+1}^{l} \right) \right], \nonumber 
\end{align} where \eqref{AppendixA_eq5} is from chain rule and Markovian assumption, and \eqref{AppendixA_eq7} is due to stationarity.

\subsection{Proof of Theorem~\ref{theorem1}}
Let $g_{J,K}\!\!\left(\!\mathbf{Y}^n_{n-J}\!,\!\mathbf{X}^n_{n-K+1)} \!\right)\!\! \!= \!\!\! - \!\log \mathrm{P}\!\left(y_n| \mathbf{Y}^{n-1}_{n-J}\!,\!\mathbf{X}^n_{n-K+1} \right)$ be a fixed function over the states of the Markov chain $\left(\!\mathbf{Y}_{n-J}^{n}\!,\!\mathbf{X}_{n-K+1}^{n}\!\right)$. From the strong law of large numbers for Markov chains \cite{Meyn2009} which states that for a fixed function $g\left(.\right)$ over the states of the Markov chain, the sample mean will almost surely converge to the expected value as $N \rightarrow \infty$, we have,
\begin{align}\label{AppendixA_eq3}
\tfrac{1}{N}\!\!\textstyle \sum\limits_{n=1}^{N}\!\! g_{J,K}\!\!\left(\!\mathbf{Y}_{n-J}^{n}\!,\!\mathbf{X}_{n-K+1}^{n}\! \right)\!\! \xrightarrow{a.s.} \!\! \mathbb{E}\!\!\left[\! g_{J,K}\!\!\left(\!\mathbf{Y}_{l-J}^{l}\!,\!\mathbf{X}_{l-K+1}^{l} \!\right) \!\right].
\end{align}
We also have
\begin{align} 
h\!\!\left(\!\mathbf{Y}\!\|\!\mathbf{X}\!\right) & \!\!=\!\! \lim\limits_{\scriptscriptstyle{N \rightarrow \infty}} \tfrac{1}{N}\! h\!\!\left(\!\mathbf{Y}^N \!\|\! \mathbf{X}^N\!\right) \!\!=\!\! \lim\limits_{\scriptscriptstyle{N \rightarrow \infty}} \!\!\mathbb{E}\!\!\left[\! g_{\scriptscriptstyle{J,K}}\!\!\left(\!\mathbf{Y}_{l-J}^{l}\!,\!\mathbf{X}_{l-K+1}^{l} \!\right) \right] \label{AppendixA_eq8} \\
& = \mathbb{E}\left[ g_{J,K}\left(\mathbf{Y}_{l-J}^{l},\mathbf{X}_{l-K+1}^{l} \right) \right] \label{AppendixA_eq4},
\end{align} 
where \eqref{AppendixA_eq8} is from Lemma.~\ref{lemma2}. We have from \eqref{AppendixA_eq3}, \eqref{AppendixA_eq4}, $ \hat{h}\left(\mathbf{Y} \| \mathbf{X} \right) \!=\! \tfrac{1}{N} \textstyle \sum\limits_{n=1}^{N} g_{J,K}\left(\mathbf{Y}^n_{n-J},\mathbf{X}^n_{n-\left(K-1\right)} \right) \! \xrightarrow{a.s.} \!h\left(\mathbf{Y} \| \mathbf{X} \right). $
\section{Derivation of DI for Linear Two Node Network}\label{AppendixC}
Consider the MVAR model in section~\ref{subsec:linear_two_node}, described by \eqref{linear_two_node_eq}. Here we will derive the DI in both directions between time-series $\mathbf{X}$ and $\mathbf{Y}$ for non-zero $\beta_1, \beta_2$. Appendix.~\ref{AppendixC2} considers the case when $\left(\beta_1, \beta_2\right) \in \left\{\left(1,0\right),\left(1,0\right)\right\}$.
\subsection{DI from \textbf{X} to \textbf{Y}}
For the system described by \eqref{linear_two_node_eq}, the causal conditional entropy $h\left(\mathbf{Y}\|\mathbf{X}\right)$ is given by 
\begin{align}
h\left(\mathbf{Y}\|\mathbf{X}\right) = \lim\limits_{N \rightarrow \infty} \frac{1}{N} h\left(\mathbf{Y}^N \| \mathbf{X}^N\right) = \frac{1}{2} \log\left(2\pi e \sigma_z^2 \right), \label{AppendixCeq5}
\end{align}
because conditioned on $\left(x_n,x_{n-1}\right)$, the only uncertainty in $y_n$ is due to the i.i.d Gaussian noise $\mathbf{Z}$ of variance $\sigma_z^2$ which is independent of $\mathbf{X}$.

Now, from \eqref{linear_two_node_eq}, we have $\left(y_1,y_2,\cdots,y_N\right)^T \sim \mathcal{N}\left(\mathbf{0} , \Sigma_N \right)$, where $\Sigma_N = \delta M_{N}$ with $\delta = \beta_1 \beta_2 \sigma_x^2$. $M_N$ is a tridiagonal matrix whose main diagonal elements are $D$ and non-zero diagonal below and above the main diagonal are all 1. $D = \frac{\gamma}{\delta}$, where $\gamma = \left(\beta_1^2 + \beta_2^2 \right)\sigma_x^2 + \sigma_z^2$. 
Upon further simplification using the tridiagonal matrix determinant from \cite{Hu1996}, we have 
\begin{align}
\left|\Sigma_N \right| = \left|\delta \right|^N \frac{\sinh \left(\left(N+1\right) \lambda\right)}{\sinh \lambda}, \: \mbox{where} \: \lambda = \cosh^{-1} \left(\frac{\left|D\right|}{2} \right).
\end{align}
The unconditioned entropy of $\mathbf{Y}$ is now given by
\begin{align} 
h\left(\mathbf{Y}\right) \!\! = \!\! \lim\limits_{N \rightarrow \infty}\!\! \frac{1}{2N}\!\! \log\!\!\left(\!\!\left(\!2\pi e\!\right)^N\!\! \left|\!\Sigma_N \!\right|\!\right) \!\! = \!\! \frac{1}{2}\!\log\! \left(\!2\pi e\! \left|\delta\right|\!\right) \!\! + \!\! \frac{1}{2}\lambda, \label{AppendixCeq8}
\end{align}
obtained by expanding the hyperbolic $\sinh$ function in the determinant $\left|\Sigma_N \right|$ in terms of exponentials and some basic algebraic manipulations. Now, from \eqref{AppendixCeq5} and \eqref{AppendixCeq8}, we have
\begin{align}
I\left(\mathbf{X} \rightarrow \mathbf{Y}\right) = \frac{1}{2} \log \left(\frac{|\beta_1 \beta_2|\sigma_x^2 }{\sigma_z^2} \right) + \frac{1}{2} \cosh^{-1}\left(\frac{\left(\beta_1^2 + \beta_2^2\right)\sigma_x^2 + \sigma_z^2}{2|\beta_1 \beta_2|\sigma_x^2} \right). \nonumber
\end{align}

\subsection{DI from \textbf{Y} to \textbf{X}}
The causal conditional entropy, $h\left(\mathbf{X} \|\mathbf{Y}\right)$ is given by 
\begin{align}
& h\left(\mathbf{X} \|\mathbf{Y}\right) = \lim\limits_{N \rightarrow \infty} \frac{1}{N} \textstyle \sum\limits_{n=1}^{N} h\left(x_n| x_{n-1}, y_{n} \right) \label{AppendixCeq10} \\
& = \lim\limits_{N \rightarrow \infty} \frac{1}{N} \textstyle \sum\limits_{n=1}^{N} \left\{h\left(x_n , y_{n} , x_{n-1}\right) - h\left(x_{n-1}, y_{n} \right) \right\} \nonumber \\
& = \lim\limits_{N \rightarrow \infty} \frac{1}{N} \textstyle \sum\limits_{n=1}^{N} \left\{ \frac{1}{2} \log\left(2\pi e |\Phi_1| \right) - \frac{1}{2} \log\left(2\pi e |\Phi_2| \right)\right\} \nonumber \\
& = \frac{1}{2} \log\left(2\pi e \frac{\sigma_x^2 \sigma_z^2}{\beta_1^2 \sigma_x^2 + \sigma_z^2} \right), \label{AppendixCeq11}								
\end{align}
where $\Phi_1$ and $\Phi_2$ are the appropriate covariance matrices. 
The reason for \eqref{AppendixCeq10} is that conditioned on $x_{n-1}$ and $y_n$, $x_n$ is independent of the other past samples of $\mathbf{X}$ and $\mathbf{Y}$. Since $\mathbf{X}$ is drawn from i.i.d. Gaussian distribution with mean zero and variance $\sigma_x^2$, the unconditional entropy of $\mathbf{X}$ is given by
$h\left(\mathbf{X}\right) = \frac{1}{2} \log \left(2\pi e \sigma_x^2 \right). $
Therefore, the DI from $\mathbf{Y}$ to $\mathbf{X}$ is
\begin{align}
I\left(\mathbf{Y} \rightarrow \mathbf{X}\right)\!\! = \!\!h\left(\mathbf{X}\right) - h\left(\mathbf{X}\|\mathbf{Y}\right)\!\! = \!\! \frac{1}{2} \log \left(1 + \frac{\beta_1^2\sigma_x^2}{\sigma_z^2} \right).
\end{align}

\subsection{Special cases} \label{AppendixC2}
Consider the system in \eqref{linear_two_node_eq} with $\beta_1=1,\beta_2=0$. For this system, $y_n$ are i.i.d. Gaussian distributed with mean zero and variance $\left(\sigma_x^2 + \sigma_z^2\right)$. Therefore the differential entropy of $\mathbf{Y}_1^N$ is given by $h\left(\mathbf{Y}^N\right) = \frac{N}{2}\log\left(2\pi e \left(\sigma_x^2 + \sigma_z^2 \right) \right)$. Also the joint differential entropy of $x_n$ and $y_n$ is 
\begin{align}
h\!\left(x_n,y_n\right) \!\! =  \!\! \textstyle \frac{1}{2}\!\log \! \left(\!2\pi e \begin{vmatrix}
														\sigma_x^2 \!\!& \!\!\sigma_x^2 \\
														\sigma_x^2 \!\!& \!\!\sigma_x^2 + \sigma_z^2
														\end{vmatrix}  \right) 
					\!\! = \!\! \frac{1}{2}\!\log \!\left(\!2\pi e 	\sigma_x^2	\sigma_z^2\! \right).											
\end{align}
$ \implies h\left(\mathbf{Y}^N\|\mathbf{X}^N\right)  = \textstyle \textstyle \sum\limits_{n=1}^N h\left(y_n|x_n\right)$ $ = \textstyle\textstyle \sum\limits_{n=1}^N \big(h\left(x_n,y_n\right) - h\left(x_n\right) \big) = \frac{N}{2}\log \left(2\pi e \sigma_z^2 \right). $
Therefore the directed information from $\mathbf{X}$ to $\mathbf{Y}$ is given by
\begin{align}
I\!\!\left(\mathbf{X} \!\! \rightarrow \!\! \mathbf{Y}\right)\!\! = \!\! \lim\limits_{N \! \rightarrow \! \infty} \!\! \big(\! h\left(\!\mathbf{Y}^N\!\right) \!\!-\!\! h\left(\!\mathbf{Y}^N\!\|\!\mathbf{X}^N\!\right)\!\!\big) \!\! = \!\! \frac{1}{2}\!\log\! \left(\!1\!+\!\frac{\sigma_x^2}{\sigma_z^2}\!\right).
\end{align}
The DI from $\mathbf{Y}$ to $\mathbf{X}$ can be similarly derived.

Now, consider the system in \eqref{linear_two_node_eq} with $\beta_1=0,\beta_2=1$. For this system, the DI from $\mathbf{X}$ to $\mathbf{Y}$ is computed by following the approach described above. Let us derive $I\left(\mathbf{Y}\rightarrow \mathbf{X}\right)$. The causal conditional entropy of $\mathbf{X}^N$ given $\mathbf{Y}^N$ is given by
\begin{equation} \label{AppendixCeq3}
h\!\left(\mathbf{X}^N \!\|\! \mathbf{Y}^N \right) \!\!=\!\!\! \textstyle \sum\limits_{n=1}^{N}\!\! h\!\left(x_n| \mathbf{X}^{n-1}\!,\!\mathbf{Y}^N\right) 
					\!\! =\!\!\!\textstyle \sum\limits_{n=1}^{N}\!\! h\!\left(x_n\right)\!\! = \!\!h\!\left( \mathbf{X}^N\right),
\end{equation}
since $x_n$ does not depend on the past samples of $\mathbf{Y}$. Therefore, the DI from $\mathbf{Y}$ to $\mathbf{X}$ is zero, i.e., $I\left(\mathbf{Y} \rightarrow \mathbf{X} \right) = 0$.
\end{appendices}

\bibliographystyle{IEEEtran}
\bibliography{refs}

\begin{wrapfigure}{L}{0.28\columnwidth}
\centering
\includegraphics[width=0.3\columnwidth]{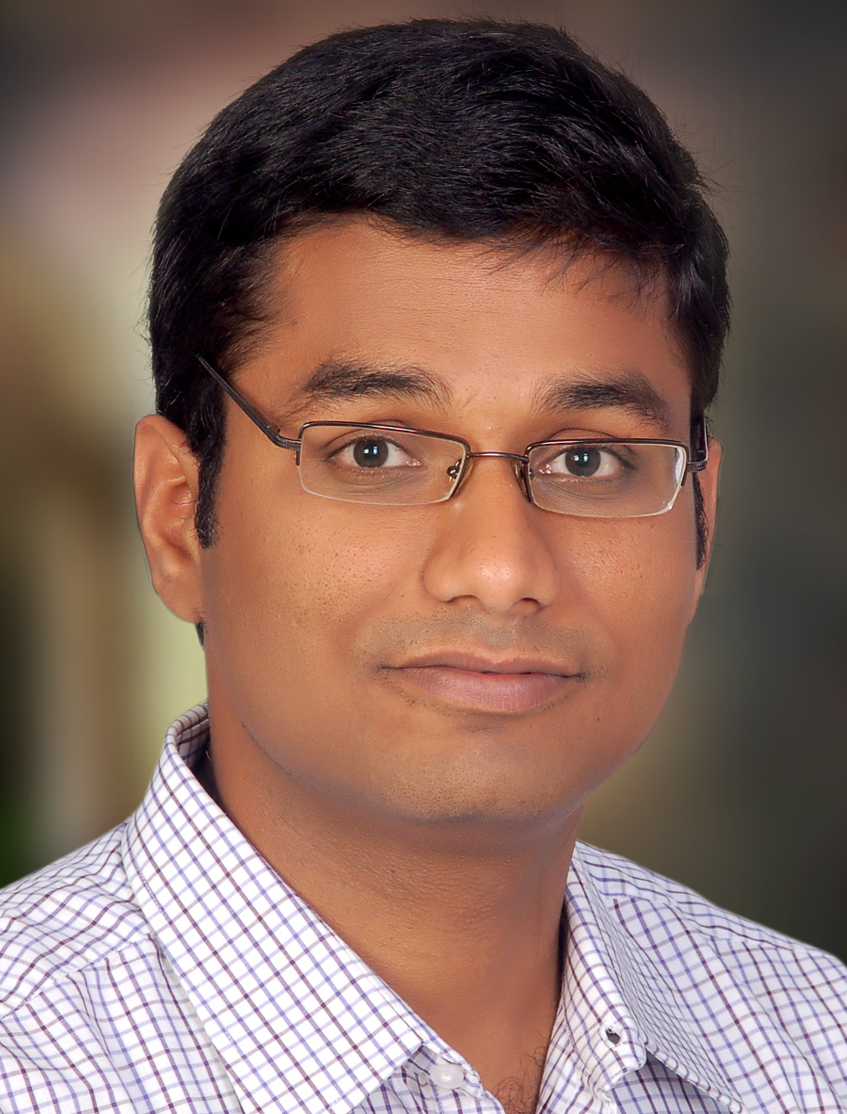}
\end{wrapfigure}
Rakesh Malladi is a graduate student in Electrical and Computer Engineering Department at Rice University,
Houston since fall 2011. He graduated with a dual degree (B.Tech $\&$ M.Tech) in Electrical Engineering from Indian
Institute of Technology Madras in 2011. He was a research intern at Cyberonics, Houston, USA during the summer
of 2015, at Texas Instruments, Dallas, USA during the summer of 2013 and at IBM Research, Bangalore, India
during the summer of 2010. His research interests spans topics in signal processing, machine learning, computational
neuroscience and information theory.

\begin{wrapfigure}{L}{0.33\columnwidth}
\centering
\includegraphics[width=0.35\columnwidth]{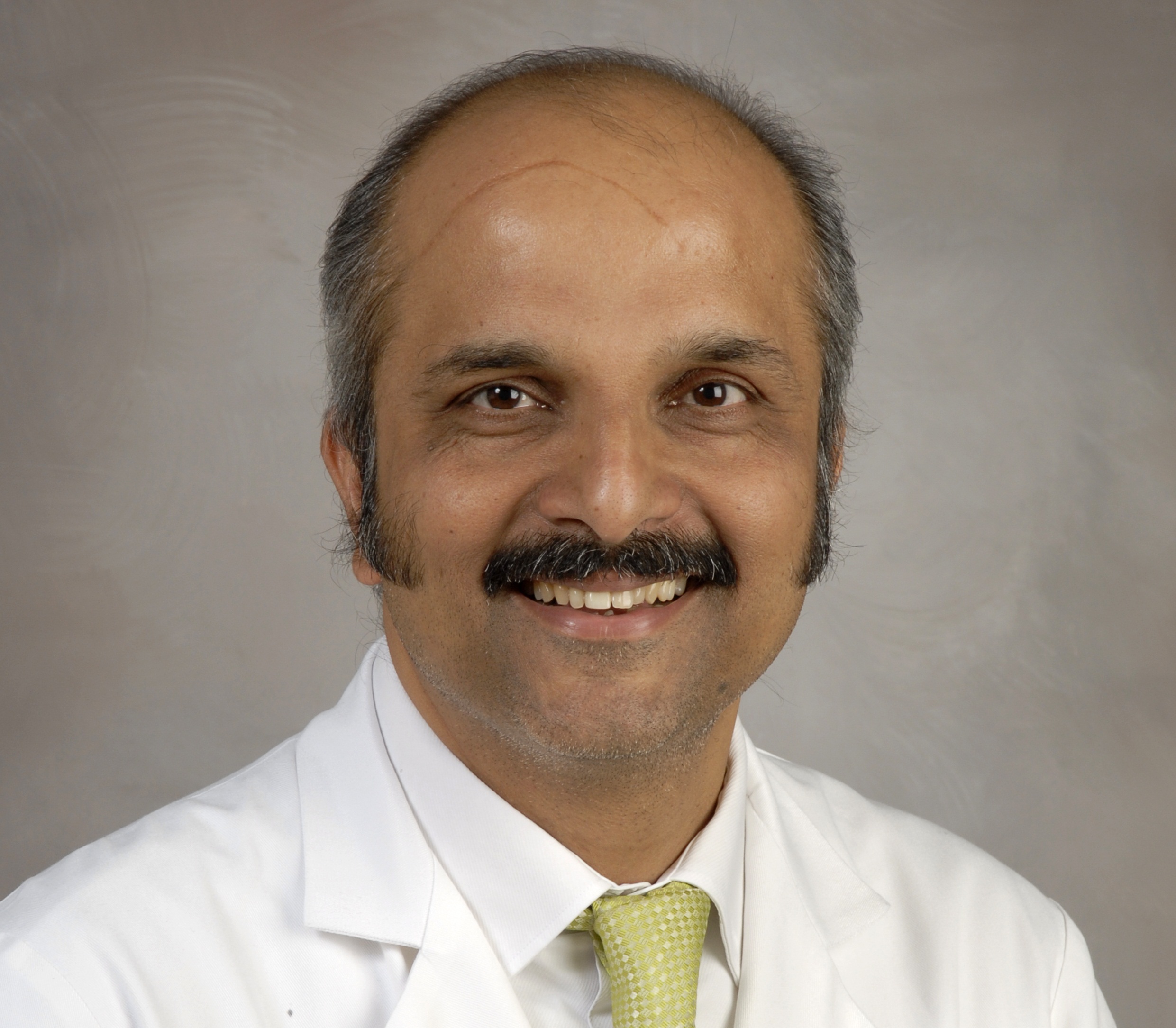}
\end{wrapfigure}
Giridhar Kalamangalam was born in Salem, India in 1965. He received the MBBS degree from the Jawaharlal Institute in Pondicherry, India in 1989, and the MSc and DPhil degrees in applied mathematics from Oxford University, UK, in 1991 and 1995. Following postgraduate general medical training (MRCP (UK)) he qualified as a neurologist (Glasgow, UK) with subspecialty expertise in epilepsy (Cleveland Clinic, USA). Since 2006, he has been with the University of Texas Health Science Center in Houston, TX, where he is currently Associate Professor of Neurology. His research interests are in the physiological dynamics and the neuroimaging of epilepsy and cognitive function.

\begin{wrapfigure}{L}{0.28\columnwidth}
\centering
\includegraphics[width=0.3\columnwidth]{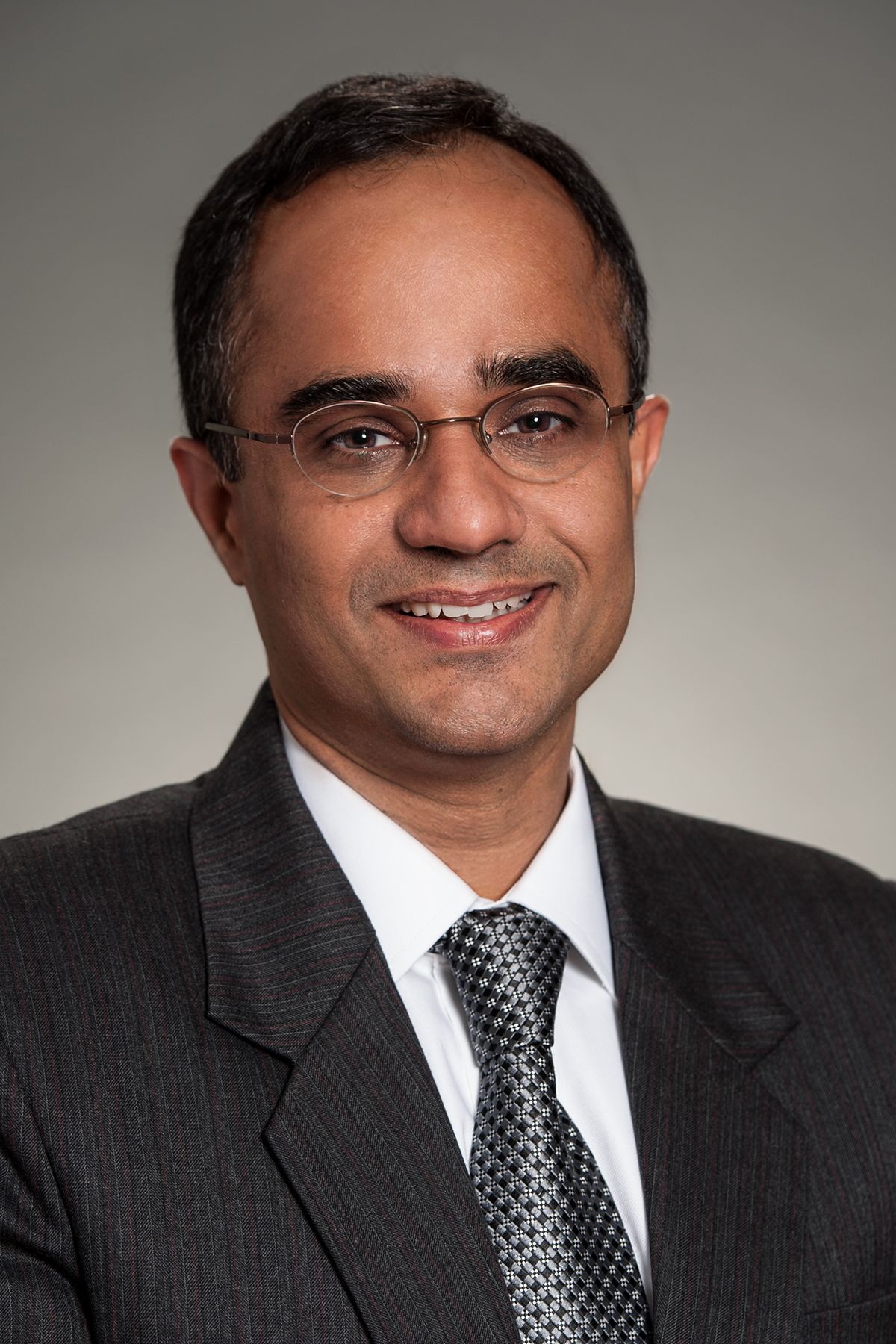}
\end{wrapfigure}
Nitin Tandon, MD, FAANS is a professor of neurosurgery at University of Texas Medical School at Houston. He received his medical degree at the Armed Forces Medical College in Pune, India, followed by a residency in neurosurgery at The University of Texas Health Science Center in San Antonio and a fellowship in epilepsy surgery at Cleveland Clinic. He has been in practice for thirteen years at Memorial Hermann TMC and on faculty at the UT Health Medical School. He has performed more than 3000 brain operations, with over 1200 for brain tumors and 600  for epilepsy. His lab performs multi-modality assessments of cognitive functions combining and correlating intracranial recordings, functional MRI, tractography and direct cortical stimulation. The thrust of his research is the development of optimal tools to characterize interactional brain processes across regions using intracranial recordings (www.tandonlab.org)

\begin{wrapfigure}{L}{0.28\columnwidth}
\centering
\includegraphics[width=0.3\columnwidth]{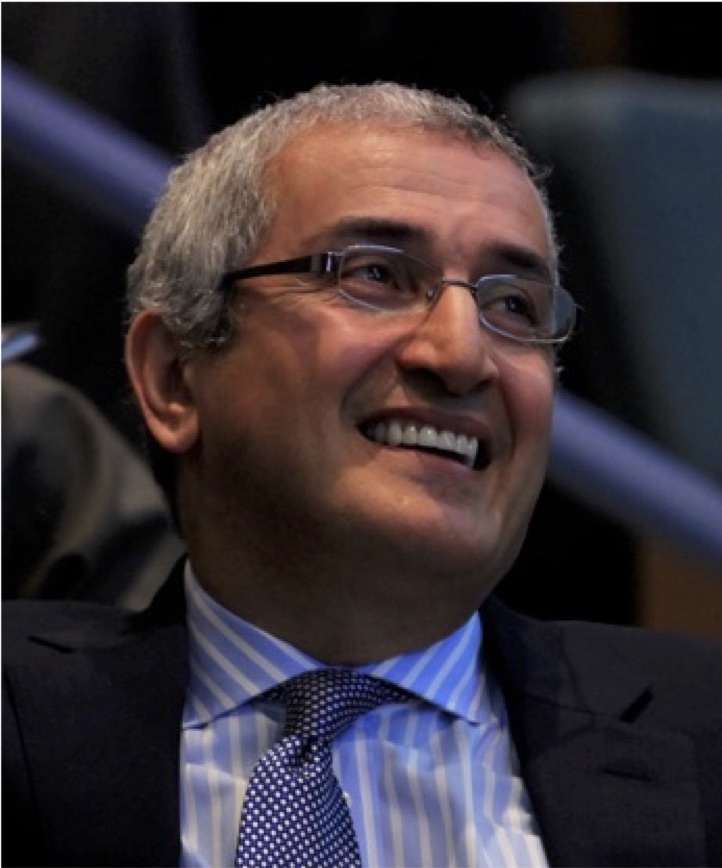}
\end{wrapfigure}
Behnaam Aazhang received his B.S. (with highest honors), M.S., and Ph.D. degrees in Electrical and Computer Engineering from University of Illinois at Urbana-Champaign in 1981, 1983, and 1986, respectively. In August 1985, he joined the faculty of Rice University, Houston, Texas, where he is now the J.S. Abercrombie Professor in the Department of Electrical and Computer Engineering Professor and Director of Center on Neuro-Engineering, a multi-university research center in Houston, Texas. He is a Fellow of IEEE and AAAS, and also a recipient of 2004 IEEE Communication Society's Stephen O. Rice best paper award for a paper with A. Sendonaris and E. Erkip. In addition, the authors received IEEE Communication Society's 2013 Advances in Communication Award for the same paper. He has been listed in the Thomson-ISI Highly Cited Researchers and has been keynote and plenary speaker of several conferences.

\end{document}